\documentclass{article}
\usepackage{amssymb}
\usepackage{amsmath}
\usepackage{amsthm}
\usepackage{color}
\usepackage[colorlinks]{hyperref}
\usepackage[a4paper]{geometry}

\newtheorem{theorem}{Theorem}[section]
\newtheorem{lemma}[theorem]{Lemma}
\newtheorem{proposition}[theorem]{Proposition}

\theoremstyle{definition}

\newtheorem{example}{Example}[section]

\theoremstyle{remark}
\newtheorem{remark}{Remark}

\title{Shortest self-orthogonal and LCD embeddings of linear codes over $\mathbb{F}_q+u\mathbb{F}_q$}
\author{Junmin An\thanks{junmin0518@sogang.ac.kr, Department of Mathematics and Institute for Mathematical and Data Sciences, Sogang University, Seoul, Korea}, Jon-Lark Kim\thanks{jlkim@sogang.ac.kr, Department of Mathematics and Institute for Mathematical and Data Sciences, Sogang University, Seoul, Korea}}
\date{}

\begin{document}
\maketitle

\begin{abstract}
This paper determines the exact lengths of shortest self-orthogonal and LCD embeddings of linear codes over $\mathbb{F}_q+u\mathbb{F}_q$. By decomposing Gram matrices over $\mathbb{F}_q+u\mathbb{F}_q$ into pairs of symmetric matrices over the finite field $\mathbb{F}_q$, the embedding problems are reduced to the congruence classification of symmetric and alternate matrices over finite fields. Complete formulas for the shortest self-orthogonal embedding length are obtained, with two distinct cases arising in both even and odd characteristic. We also show that every self-orthogonal code over $\mathbb{F}_q+u\mathbb{F}_q$ with nonzero free rank can be viewed as a shortest self-orthogonal embedding of another code. We use Witt theory to construct all shortest self-orthogonal embeddings. A complete characterization of shortest LCD embeddings is also established in terms of invertible and arbitrary matrices of prescribed sizes appended to a generator matrix. Examples of self-orthogonal and LCD embeddings with the largest minimum distance for the code considered are also presented, some of whose Gray images are optimal codes over $\mathbb{F}_q$.\\

\textbf{Keywords} : self-orthogonal code, LCD code, embedding of codes, codes over rings
\end{abstract}

\section{Introduction}
Following the observations in~\cite{FST-1992,HKCSS-1994} that several good nonlinear binary codes are closely related to linear codes over $\mathbb Z_4$ via the Gray map, codes over $\mathbb Z_4$ have been studied extensively. As the ring $\mathbb{F}_2+u\mathbb{F}_2$ shares some properties of $\mathbb{Z}_4$ and the finite field $\mathbb{F}_4$, codes over this ring and, more generally, over $\mathbb{F}_q+u\mathbb{F}_q$ have also played an important role in the study of codes over rings. Codes over $\mathbb{F}_q+u\mathbb{F}_q$ were first studied in~\cite{B-1997} with an application to modular lattices. In particular, self-dual codes over $\mathbb{F}_q+u\mathbb{F}_q$ have been extensively studied~\cite{BU-1999,DGHMS-1999,DGHS-1999,G-1996,H-2007} with classical results for self-dual codes over finite fields including invariant theory, weight enumerators, and mass formulas being generalized to codes over $\mathbb{F}_q+u\mathbb{F}_q$. Similarly, well-known families of classical codes, such as cyclic codes~\cite{AS-2007,BU-1999,UB-1999}, constacyclic codes~\cite{D-2009,D-2010,DNSS-2018,DNSY-2021}, and duadic codes~\cite{L-2001} have also been generalized to $\mathbb{F}_q+u\mathbb{F}_q$.

Recently, embedding methods have been studied as an approach to constructing new codes with good minimum distances. These methods construct a new code by appending columns to a generator matrix of a given code. A central problem in the study of embedding methods is to determine the minimum number of columns that must be appended so that the resulting code has a certain property, such as being self-orthogonal or LCD, or having a specified hull dimension. Kim et al.~\cite{KKL-2021} first investigated this problem. They provided an algorithm to obtain binary self-orthogonal codes by adding the minimum number of columns to a generator matrix of a linear code with dimension three and four. This result was extended to dimension five and six by Kim and Choi~\cite{KC-2022}. An et al.~\cite{AKKLW-arXiv1} completely determined the minimum number of columns required to embed a given binary linear code into a self-orthogonal code for all cases. In~\cite{A-2026}, the embedding approach was extended to LCD codes. Wang and Luo~\cite{WL-2026} further extended these results by completely determining the minimum number of columns required to embed a given linear code over an arbitrary finite field into a code with a prescribed hull dimension. Most recently, An et al.~\cite{AKS-arXiv} extended the embedding approach to codes over $\mathbb Z_4$. To this end, they also investigated shortest doubly even self-orthogonal embeddings of binary linear codes. The aim of this paper is to extend the embedding methods to linear codes over the ring $R=\mathbb F_q+u\mathbb F_q$.

In this paper, we study the exact length of a shortest self-orthogonal embedding and a shortest LCD embedding of linear codes over $R=\mathbb{F}_q+u\mathbb{F}_q$. Every matrix $A$ over $R$ can be uniquely decomposed into $A=A_0+uA_1$ for some matrices $A_0$ and $A_1$ over the finite field $\mathbb{F}_q$. In particular, the Gram matrix associated with a generator matrix of a linear code over $R$ decomposes into two symmetric matrices over $\mathbb{F}_q$. This allows us to solve the shortest self-orthogonal and LCD embedding length problems using the congruence theory of symmetric and alternate matrices over finite fields. Consequently, we determine the exact length of a shortest self-orthogonal embedding for linear codes over $\mathbb F_q+u\mathbb F_q$ in all cases. In particular, we show that every self-orthogonal code over $\mathbb{F}_q+u\mathbb{F}_q$ whose free rank is not zero is a shortest self-orthogonal embedding of other codes with smaller lengths. Furthermore, using the Witt theory over finite fields, we provide a construction method for obtaining all shortest self-orthogonal embeddings. For each of the cases where $q$ is even or odd, the length of a shortest self-orthogonal embedding splits into two cases. We provide an example for each of these four cases. Applying the construction method described above, we enumerate all shortest self-orthogonal embeddings in each example and exhibit one having the largest minimum distance among them. In some of these examples, the Gray images of the resulting shortest self-orthogonal embeddings are optimal codes over $\mathbb{F}_q$. For shortest LCD embeddings, we show that all such embeddings of a given code can be obtained by appending an arbitrary invertible matrix and an arbitrary matrix of prescribed sizes to its generator matrix. As in the self-orthogonal case, we apply this characterization to enumerate all shortest LCD embeddings and provide an example of finding one having the largest minimum distance among them.

This paper is organized as follows. In Section 2, we provide preliminaries on the ring $R=\mathbb{F}_q+u\mathbb{F}_q$ and codes over $R$. We also briefly review the Witt theory over finite fields. In Section 3, we present our results on shortest self-orthogonal embeddings of linear codes over $R$ with even characteristic. By dividing the codes into two cases, we classify the exact length of the shortest self-orthogonal embeddings for each case. In Section 4, we establish the corresponding results for linear codes over $R$ of odd characteristic, where the codes are likewise divided into two cases and the exact length of a shortest self-orthogonal embedding is determined in each case. In Section 5, we develop a method for constructing all possible shortest self-orthogonal embeddings. In Section 6, we present our results on shortest LCD embeddings. In Section 7, we provide examples illustrating each case of shortest self-orthogonal and shortest LCD embeddings. We conclude our paper in Section 8.

\section{Preliminaries}
We refer the reader to~\cite{book-HP,book-MS} for a detailed background on coding theory.

Let $\mathbb{F}_q$ be a finite field of order $q$ where $q$ is a prime power. The ring $R=\mathbb{F}_q+u\mathbb{F}_q$ is defined as
\[
R=\mathbb{F}_q+u\mathbb{F}_q=\mathbb{F}_q[u]/(u^2).
\]
The ring $R$ is a commutative local ring with maximal ideal $(u)$. Thus, a square matrix $A$ over $R$ is invertible if and only if $A$ modulo $(u)$ is invertible over $R/(u)\cong\mathbb{F}_q$.

A \textit{code} $\mathcal{C}$ of length $n$ over $R$ is a subset of $R^n$ and its elements are called \textit{codewords}. If $\mathcal{C}$ is an $R$-submodule of $R^n$, then we call $\mathcal{C}$ a \textit{linear code} over $R$. Throughout this paper, a code over $R$ is always assumed to be linear unless otherwise stated. A linear code $\mathcal{C}$ of length $n$ over $R$ contains a set of $k_1+k_2$ codewords $\mathbf{c}_1, \ldots, \mathbf{c}_{k_1}\in R^n\backslash(uR)^n$ and $\mathbf{c}_{k_1+1}, \ldots, \mathbf{c}_{k_1+k_2}\in (uR)^n$ such that every codeword in $\mathcal{C}$ has a unique representation of the form
\[
\sum_{i=1}^{k_1}a_i\mathbf{c}_i+\sum_{j=k_1+1}^{k_1+k_2}b_j\mathbf{c}_j
\]
for some $a_i\in R$ and $b_j\in \mathbb{F}_q$. A \textit{generator matrix} $G$ for $\mathcal{C}$ is a matrix whose rows are $\mathbf{c}_i$ for $1\le i\le k_1+k_2$. Here $\mathcal{C}$ is said to be of \textit{type $\{k_1,\,k_2\}$}. We call $k_1$ the \textit{(free) rank} of $\mathcal{C}$. If $k_2=0$, then $\mathcal{C}$ is called a \textit{free code}. For a code $\mathcal{C}$ over $R$, we denote a generator matrix for $\mathcal{C}$ by $G(\mathcal{C})$. Conversely, for a matrix $A$ over $R$, we denote the code generated by the rows of $A$ by $\mathcal{C}_A$.

Two codes $\mathcal{C}_1$ and $\mathcal{C}_2$ over $R$ are said to be \textit{(permutation) equivalent} if there is a permutation of columns $\sigma$ such that $\sigma\mathcal{C}_1=\mathcal{C}_2$. Up to equivalence, every code of type $\{k_1,\,k_2\}$ over $R$ has a generator matrix of the form
\[
\begin{pmatrix}
I_{k_1} & A & B_1+uB_2\\\mathcal{O} & uI_{k_2} & uD
\end{pmatrix}
\]
where $A,\,B_1,\,B_2$ and $D$ are matrices over $\mathbb{F}_q$. Such a matrix is called the \textit{standard generator matrix} of $\mathcal{C}$.

For a code $\mathcal{C}$ of length $n$ over $R$, the \textit{dual} $\mathcal{C}^\perp$ of $\mathcal{C}$ is defined as
\[
\mathcal{C}^\perp=\{\mathbf{x}\in R^n~|~\mathbf{x}\cdot \mathbf{c}=0~\mbox{for all }\mathbf{c}\in \mathcal{C}\}
\]
where $\cdot$ is the standard dot product. If the type of $\mathcal{C}$ is $\{k_1,\,k_2\}$, then the type of $\mathcal{C}^\perp$ is $\{n-k_1-k_2,\,k_2\}$. If $\mathcal{C}\subseteq\mathcal{C}^\perp$, then $\mathcal{C}$ is called a \textit{self-orthogonal code}. Moreover, if $\mathcal{C}=\mathcal{C}^\perp$, then it is called a \textit{self-dual code}. The \textit{hull} of $\mathcal{C}$ is defined as
\[
\operatorname{Hull}(\mathcal{C})=\mathcal{C}\cap\mathcal{C}^\perp.
\]
If $\mathcal{C}=\operatorname{Hull}(\mathcal{C})$, then $\mathcal{C}$ is self-orthogonal. A code $\mathcal{C}$ is called an \textit{LCD (linear complementary dual) code} if $\operatorname{Hull}(\mathcal{C})=\{0\}$. Let $G$ be a generator matrix of a code $\mathcal{C}$ over $R$. It is easy to check that $\mathcal{C}$ is self-orthogonal if and only if $GG^T=\mathcal{O}$. Bhowmick et al.~\cite{BFMBB-2020} showed that any LCD code over a finite commutative local Frobenius ring, and hence any LCD code over $R$, is free. They also showed that a free code over $R$ is LCD if and only if $GG^T$ is invertible.

The \textit{(Hamming) weight} of a vector $\mathbf{u}\in R^n$ is the number of nonzero coordinates in $\mathbf{u}$, and is denoted by $\operatorname{wt}(\mathbf{u})$. For two vectors $\mathbf{u}$ and $\mathbf{v}$ in $R^n$, the \textit{(Hamming) distance} between $\mathbf{u}$ and $\mathbf{v}$ is defined as
\[
d(\mathbf{u},\,\mathbf{v})=\operatorname{wt}(\mathbf{u}-\mathbf{v}).
\]
The \textit{minimum (Hamming) distance} of a code $\mathcal{C}$ over $R$ is the minimum of the distances between any two distinct codewords in $\mathcal{C}$.

The $\mathbb{F}_q$-linear map $\phi:R\to \mathbb{F}_q^2$ defined as
\[
\phi(a+ub)=(b, a+b)
\]
for $a, b\in \mathbb{F}_q$ is called the \textit{Gray map}. It extends componentwise to a map on $R^n$ as
\[
\phi(x_1,\ldots,x_n)
=\bigl(\phi(x_1),\ldots,\phi(x_n)\bigr).
\]
If $q$ is even, then Gray map preserves orthogonality, that is, it maps a self-orthogonal code over $R$ into a self-orthogonal code over $\mathbb{F}_q$.

For a code $\mathcal{C}$ over $R$, the \textit{residue code} of $\mathcal{C}$, denoted by $\text{Res}(\mathcal{C})$, is a code over $\mathbb{F}_q$ defined as
\[
\operatorname{Res}(\mathcal{C})=\mathcal{C}\pmod u.
\]

Let $\mathcal{C}$ be a code over $R$ with generator matrix $G$. A \textit{self-orthogonal embedding} $\tilde{\mathcal{C}}$ of $\mathcal{C}$ is a self-orthogonal code over $R$ generated by
\[
\tilde{G}=[G~|~B]
\]
for some matrix $B$ over $R$. A \textit{shortest self-orthogonal embedding} of $\mathcal{C}$ is a self-orthogonal embedding of $\mathcal{C}$ whose length is smallest among all self-orthogonal embeddings of $\mathcal{C}$. We denote the length of a shortest self-orthogonal embedding of $\mathcal{C}$ by $n_s(\mathcal{C})$.

Similarly, an \textit{LCD embedding} $\hat{\mathcal{C}}$ of $\mathcal{C}$ is an LCD code over $R$ generated by
\[
\hat{G}=[G~|~D]
\]
for some matrix $D$ over $R$. A \textit{shortest LCD embedding} of $\mathcal{C}$ is an LCD embedding of $\mathcal{C}$ whose length is smallest among them. The length of a shortest LCD embedding of $\mathcal{C}$ is denoted by $n_\ell(\mathcal{C})$.

For shortest self-orthogonal and shortest LCD embeddings of linear codes over finite fields, we refer the reader to~\cite{A-2026,AKKLW-arXiv1,WL-2026}.

Next, we briefly review symmetric bilinear forms and some basic notions from Witt theory that will be used throughout this paper. 

Let $V$ be a finite-dimensional vector space over a finite field $\mathbb{F}_q$, and $B:V\times V\to \mathbb{F}_q$ be a symmetric bilinear form on $V$. The \textit{radical} of $(V, B)$ is defined as
\[
\operatorname{Rad}(B)=\{\mathbf{u}\in V~|~B(\mathbf{u}, \mathbf{v})=0\mbox{ for all }\mathbf{v}\in V\}.
\]
If $\operatorname{Rad}(B)=\{0\}$, that is, if $B$ is non-degenerate on $V$, we call $(V, B)$ \textit{non-singular}. A subspace $W$ of $V$ is called \textit{non-singular} if the restriction $B_{W\times W}$ is non-degenerate. If $V$ is non-singular and $W$ is a non-singular subspace of $V$, then
\[
V=W\perp W^\perp
\]
where $W^\perp=\{\mathbf{u}\in V~|~B(\mathbf{u}, \mathbf{v})=0\mbox{ for all }\mathbf{v}\in W\}$. For two vector spaces $V_1,\,V_2$ equipped with symmetric bilinear forms $B_1$ and $B_2$, respectively, let $\varphi:V_1\to V_2$ be a linear isomorphism. If
\[
B_2(\varphi(\mathbf{u}),\,\varphi(\mathbf{v}))=B_1(\mathbf{u},\,\mathbf{v})
\]
for every $\mathbf{u},\,\mathbf{v}\in V_1$, then $\varphi$ is called an \textit{isometry}.

Every symmetric matrix $A\in M_n(\mathbb{F}_q)$ determines a symmetric bilinear
form $B_A$ on $\mathbb{F}_q^n$ by
\[
B_A(\mathbf{u},\mathbf{v})=\mathbf{u}A\mathbf{v}^T.
\]
Here, we define the \textit{radical} of $A$ as $\operatorname{Rad}(A)=\operatorname{Rad}(B_A)$. In this case,
\[
\operatorname{Rad}(A)=\ker A=\{\mathbf{x}\in \mathbb{F}_q^n~|~\mathbf{x}A=\mathbf{0}\}.
\]
Therefore, $(\mathbb{F}_q^n,\,B_A)$ is non-singular if and only if $A$ is invertible. Throughout this paper, we identify a symmetric matrix $A$ with its corresponding symmetric bilinear form $B_A$ whenever no confusion arises.

A bilinear form $B$ on $V$ is called \textit{alternate} if $B(\mathbf{u}, \mathbf{u})=0$ for every $\mathbf{u}\in V$. Otherwise, it is called \textit{non-alternate}. If $(V, B)$ is non-singular and $B$ is alternate, then we call $(V, B)$ a \textit{symplectic space}. A basis $e_1,\,f_1,\ldots,\,e_t,\,f_t$ of a symplectic space is called a \textit{symplectic basis} if
\[
B(e_i,\,e_j)=B(f_i,\,f_j)=0\quad\mbox{and}\quad B(e_i,\,f_j)=\delta_{ij}.
\]
A matrix $A$ is called \textit{alternate} if $A^T=-A$ and the diagonal elements of $A$ are all zero. It is easy to check that if $A$ is alternate, then $B_A$ is also alternate. Let $F$ be a field of characteristic two. Then alternate matrices over $F$ are precisely the symmetric matrices with zero diagonal. The following two theorems provide a classification of symmetric matrices over $\mathbb{F}_q$.

For an even integer $m$, define
\[
    J_m=
    \mathrm{diag}\left(
    \underbrace{
    \begin{pmatrix}
    0&1\\
    1&0
    \end{pmatrix},
    \ldots,
    \begin{pmatrix}
    0&1\\
    1&0
    \end{pmatrix}
    }_{m/2\text{ copies}}
    \right).
\]

\begin{theorem}\cite{A-1938}\label{cong_2}
    Let $A$ be a symmetric matrix of rank $r$ over $\mathbb{F}_q$ where $q$ is even. If $A$ is alternate, then $r$ is even and $A$ is congruent to $J_r\oplus\mathcal{O}$. Otherwise, if $A$ is non-alternate, then $A$ is congruent to $I_r\oplus \mathcal{O}$.
\end{theorem}

\begin{theorem}\cite{G-2010}\label{congclass}
    Let $A$ be an invertible symmetric matrix of rank $r$ over $\mathbb{F}_q$ where $q$ is odd. Then $A$ is congruent to $I_r$ if $\det A\in (\mathbb{F}_q^\times)^2$ and is congruent to $I_{r-1}\oplus[z]$ for some $z\in \mathbb{F}_q^\times\backslash(\mathbb{F}_q^\times)^2$ if $\det A\in \mathbb{F}_q^\times\backslash(\mathbb{F}_q^\times)^2$.
\end{theorem}

We briefly introduce some notions related to Witt theory. For more details, we refer the reader to~\cite{L-2005}.

Let $\mathbf{v}$ be a nonzero vector in a vector space $V$ with symmetric bilinear form $B$. Then $\mathbf{v}$ is called an \textit{isotropic vector} if $B(\mathbf{v},\,\mathbf{v})=0$, and is called an \textit{anisotropic vector} otherwise. $(V, B)$ is said to be \textit{isotropic} if it contains an isotropic vector, and is said to be \textit{anisotropic} otherwise. If all nonzero vectors of $(V, B)$ are isotropic, then it is said to be \textit{totally isotropic}. A totally isotropic subspace $W$ of $V$ is called a \textit{maximal totally isotropic subspace} if there is no totally isotropic subspace of $V$ that properly contains $W$.

\begin{theorem}\cite{L-2005}\label{Witt_basis}
Let $F$ be a field such that $\operatorname{char}F\ne 2$ and let $V$ be a non-singular vector space over $F$ equipped with a symmetric bilinear form $B$. Let $(U, B|_U)$ be a totally isotropic subspace of dimension $d$ with basis $u_1, \ldots, u_d$. Then there exists a totally isotropic subspace $(W, B|_W)$ of the same dimension with basis $v_1, \ldots, v_d$ such that
\[
B(u_i,\,v_j)=\delta_{ij}
\]
where $\delta_{ij}$ is the Kronecker delta.
\end{theorem}

\begin{theorem}[Witt's Extension Theorem~\cite{P-1964}]\label{witt_e}
    Let $V$ be a vector space over a finite field $F$ with symmetric bilinear form $B$ on $V$. Let $B$ be non-singular, non-alternate, and let $U_1$ and $U_2$ be subspaces of $V$. Let $\varphi:(U_1, B|_{U_1})\to (U_2, B|_{U_2})$ be an isometry. If $\operatorname{char}(F)\ne 2$, then $\varphi$ extends to an isometry of $V$. Otherwise, if $\operatorname{char}(F)= 2$, then $\varphi$ extends to an isometry of $V$ if and only if the following holds:
    \begin{enumerate}
        \item $h\in U_1$ if and only if $h\in U_2$, where $h=\sum_ie_i$ for an orthonormal basis $\{e_i\}_i$ of $V$.
        \item If $h\in U_1$, then $\varphi(h)=h$.
    \end{enumerate}
\end{theorem}

Let $V$ be a two-dimensional vector space over a field $F$ with $\operatorname{char}F\ne 2$, and $B$ be a symmetric bilinear form on $V$. If $(V, B)$ is non-singular and isotropic, then it is called the \textit{hyperbolic plane}. The orthogonal sum of hyperbolic planes is called \textit{hyperbolic}.

\begin{theorem}[Witt's Decomposition Theorem]\label{Witt}
Let $V$ be a vector space over a field $F$ with odd characteristic and $B$ be a symmetric bilinear form on $V$. Then $V$ splits into an orthogonal sum
\begin{equation}\label{Wittdecom}
(V_t, B|_{V_t})\perp (V_h, B|_{V_h})\perp (V_a, B|_{V_a})
\end{equation}
where $V_t$ is totally isotropic, $V_h$ is hyperbolic and $V_a$ is anisotropic.
\end{theorem}

If $V_h$ is the orthogonal sum of $w$ hyperbolic planes in the decomposition of Equation~\ref{Wittdecom}, then $w$ is called the \textit{Witt index} of $(V, B)$. 

\begin{theorem}
Let $F$ be a field such that $\operatorname{char}F\ne 2$ and let $V$ be a non-singular vector space over $F$ equipped with a symmetric bilinear form $B$. Then the Witt index $w$ of $V$ equals the dimension of any maximal totally isotropic subspace of $V$.
\end{theorem}

\section{Shortest self-orthogonal embeddings of linear codes over {\it R} with even characteristic}\label{sec3}
In this section, we assume that $R=\mathbb{F}_q+u\mathbb{F}_q$ where $q=2^s$ for some $s\ge 1$ otherwise state.

\begin{proposition}\label{soexist}
    Let $\mathcal{C}$ be a code of length $n$ over $R=\mathbb{F}_q+u\mathbb{F}_q$, where $q$ is even, and suppose that $\mathcal{C}$ is not self-orthogonal. Then $n_s(\mathcal{C})\le 2n$.
\end{proposition}
\begin{proof}
    For a generator matrix $G$ of $\mathcal{C}$, define $G'=[G~|~G]$. Since
    \[
    G'(G')^T=GG^T+GG^T=\mathcal{O},
    \]
    $G'$ generates a self-orthogonal embedding of $\mathcal{C}$.
\end{proof}
This Proposition tells us that any linear code over $R$ whish is not self-orthogonal has a self-orthogonal embedding.

For a $k\times k$ symmetric matrix $A$ over a ring $R$ (not necessarily $\operatorname{char}R= 2$ here), define $\nu(A)$ by
\[
\nu(A) = \begin{cases} 
0, & \text{if }A=\mathcal{O},\\
\min \{ m\in\mathbb{Z}_{> 0} \mid A + BB^T = \mathcal{O}, \, B \in M_{k \times m}(R) \}, & \text{if such } B \text{ exists,} \\ 
\infty, & \text{otherwise.} 
\end{cases}
\]
For a linear code $\mathcal{C}$ of length $n$ over $R$, let $G$ be a generator matrix of $\mathcal{C}$. Then $\nu(GG^T)$ is the minimal number of columns required to embed $\mathcal{C}$ into a self-orthogonal code, i.e.,
\[
n_s(\mathcal{C})-n=\nu(GG^T).
\]
Also, Proposition~\ref{soexist} show that $\nu(GG^T)$ is always finite.

We note that Propositions~\ref{embed-cong} and~\ref{embed-inequal} below do not require the
assumption that $\operatorname{char}R=2$.

\begin{proposition}\label{embed-cong}
    Let $A$ and $A'$ be congruent matrices over a ring $R$. Then $\nu(A)=\nu(A')$.
\end{proposition}
\begin{proof}
    Let $A$ be a $k\times k$ symmetric matrix over $R$ such that $\nu(A)=m<\infty$. Then, there is a matrix $B\in M_{k\times m}(R)$ such that $A+BB^T=\mathcal{O}$. Let $U$ be an invertible matrix over $R$ such that $A'=UAU^T$ and let $B'=UB$. Then we have
    \[
    A'+B'(B')^T=UAU^T+B'(B')^T=UAU^T+UBB^TU^T=U(A+BB^T)U^T=\mathcal{O}.
    \]
    It follows that $\nu(A')\le \nu(A)$. Conversely, since $U$ is invertible, we have $A=U^{-1}A'(U^{-1})^T$. By applying the same argument to $A'$ and $U^{-1}$, we obtain $\nu(A)\le \nu(A')$.  Therefore, if one of $\nu(A)$ and $\nu(A')$ is finite, then so is the other, and the two inequalities above imply
    \[
    \nu(A)=\nu(A').
    \]
    If neither is finite, then both are equal to $\infty$.

\end{proof}

\begin{proposition}\label{embed-inequal}
    Let $M$ be a $k\times k$ symmetric matrix over a ring $R$ such that 
    \[
    M=\begin{pmatrix} A & \mathcal{O}\\\mathcal{O} & B \end{pmatrix}
    \]
    for some $A\in M_{k_1}(R)$ and $B\in M_{k_2}(R)$ where $k_1+k_2=k$. Then $\nu(M)\le \nu(A)+\nu(B)$.
\end{proposition}
\begin{proof}
If either $\nu(A)$ or $\nu(B)$ is infinite, then it is trivial. Suppose that both $\nu(A)$ and $\nu(B)$ are finite. Take $X\in M_{k_1\times \nu(A)}(R)$ and $Y\in M_{k_2\times \nu(B)}(R)$ such that $A+XX^T=\mathcal{O}$ and $B+YY^T=\mathcal{O}$, respectively. Then we have \[
    M+\begin{pmatrix} X & \mathcal{O}\\\mathcal{O} & Y \end{pmatrix}\begin{pmatrix} X & \mathcal{O}\\\mathcal{O} & Y \end{pmatrix}^T=M-M=\mathcal{O}.
    \]
    Therefore, $\nu(M)\le \nu(A)+\nu(B)$.
\end{proof}

\begin{lemma}\label{nunonzerodiag}
    Let $R=\mathbb{F}_q+u\mathbb{F}_q$ with $q$ even. For an integer $m_1\ge 0$ and an even integer $m_2\ge 0$,
    \[
    \nu(I_{m_1}\oplus uJ_{m_2})=m_1+m_2
    \]
    where $I_{m_1}$ is the $m_1\times m_1$ identity matrix over $R$.
\end{lemma}
\begin{proof}
We first show that $\nu(I_{m_1}\oplus uJ_{m_2})\le m_1+m_2$. Note that
\[
\begin{pmatrix}
    0 & u\\u & 0
\end{pmatrix}=
\begin{pmatrix}
    1 & 1\\u & 0
\end{pmatrix}
\begin{pmatrix}
    1 & 1\\u & 0
\end{pmatrix}^T.
\]
Let
\[
A=I_{m_1}\oplus
\mathrm{diag}\left(
\underbrace{
\begin{pmatrix}
1&1\\
u&0
\end{pmatrix},
\ldots,
\begin{pmatrix}
1&1\\
u&0
\end{pmatrix}
}_{m_2/2\text{ copies}}
\right).
\]
Since $I_{m_1}\oplus uJ_{m_2}=AA^T$, it follows that $\nu(I_{m_1}\oplus uJ_{m_2})\le m_1+m_2$.

Next, we show that $\nu(I_{m_1}\oplus uJ_{m_2})\ge m_1+m_2$. Let $\nu(I_{m_1}\oplus uJ_{m_2})=m$ and $n=m_1+m_2$. Then there exists $B\in M_{n\times m}(R)$ such that
\[
I_{m_1}\oplus uJ_{m_2}=BB^T.
\]
Let $\mathcal{C}_B$ be a code generated by the rows of $B$ with type $\{a, b\}$. Then $a+b\le m$. Define a map $\varphi:R^{m}\to R^n$ as
\[
\varphi(\mathbf{x})=\mathbf{x}B^T
\]
for $\mathbf{x}\in R^m$. Then we have
\[
\varphi(\mathcal{C}_B)=\langle BB^T\rangle=\langle I_{m_1}\oplus uJ_{m_2}\rangle=R^{m_1}\oplus (uR)^{m_2}.
\]
Since the minimum number of generators of the image of an $R$-linear map is less than or equal to that of its domain, we have
\[
m_1+m_2\le a+b\le m=\nu(I_{m_1}\oplus uJ_{m_2}).
\]
Thus, we conclude that $\nu(I_{m_1}\oplus uJ_{m_2})=m_1+m_2$.
\end{proof}

\begin{lemma}\label{nuzerodiag}
    Let $R=\mathbb{F}_q+u\mathbb{F}_q$ with $q$ even. For even integers $m_1> 0$ and $m_2\ge 0$,
    \[
    \nu(J_{m_1}\oplus uJ_{m_2})=\begin{cases}
        m_1+m_2, & \mbox{if }m_2>0,\\
        m_1+1, & \mbox{if }m_2=0.
    \end{cases}
    \]
\end{lemma}
\begin{proof}
We first show that
\[
\nu(J_{m_1}\oplus uJ_{m_2})\le\begin{cases}
        m_1+m_2, & \mbox{if }m_2>0,\\
        m_1+1, & \mbox{if }m_2=0.
    \end{cases}
\]
Let $m_2=0$ and $t=m_1/2$. Consider the even weight code $E_{m_1+1}$ over $\{0, 1\}\subseteq \mathbb{F}_q$ of length $m_1+1$. Since $m_1+1$ is an odd integer, $E_{m_1+1}$ does not contain the all-one vector $\mathbf{1}$. Thus,
\[
\mathrm{Hull}(E_{m_1+1})=E_{m_1+1}\cap\langle \mathbf{1}\rangle=\{0\}.
\]
It follows that the standard inner product on $E_{m_1+1}$ is a nondegenerate alternate bilinear form. Let $e_1,\, f_1, \ldots, e_s,\, f_s$ be the symplectic basis of $E_{m_1+1}$ where $s=m_1/2$. Let $A$ be the matrix whose rows are $e_1,\, f_1, \ldots, e_s,\, f_s$ in order. Viewing $A$ as a matrix over $R$, we have $AA^T=J_{m_1}$. Thus, $\nu(J_{m_1})\le m_1+1$.

Assume that $m_2>0$. Define
\[
B=\begin{pmatrix}
    A & \mathbf{0}^T\\u\mathbf{1} & 0\\\mathbf{1} & 1
\end{pmatrix}
\]
where $\mathbf{1}$ is the all-one vector of length $m_1+1$, and view $A$ as a matrix over $R$. Then
\[
BB^T=\begin{pmatrix}
    AA^T & uA\mathbf{1}^T & A\mathbf{1}^T\\u\mathbf{1}A^T & u^2\mathbf{1}\cdot\mathbf{1} & u\mathbf{1}\cdot \mathbf{1}\\\mathbf{1}A^T & u\mathbf{1}\cdot \mathbf{1} & \mathbf{1}\cdot \mathbf{1}+1
\end{pmatrix}=
\begin{pmatrix}
    J_{m_1} & \mathcal{O} & \mathcal{O}\\\mathbf{0} & 0 & u\\\mathbf{0} & u & 0
\end{pmatrix}=J_{m_1}\oplus uJ_2.
\]
Define
\[
B'=\mathrm{diag}\left(
B,\underbrace{\begin{pmatrix}
1&1\\
u&0
\end{pmatrix},
\ldots,
\begin{pmatrix}
1&1\\
u&0
\end{pmatrix}
}_{(m_2-2)/2\text{ copies}}
\right).
\]
Since $B'(B')^T=J_{m_1}\oplus uJ_{m_2}$, we have $\nu(J_{m_1}\oplus uJ_{m_2})\le m_1+m_2$.

Next, we show that
\[
\nu(J_{m_1}\oplus uJ_{m_2})\ge \begin{cases}
        m_1+m_2, & \mbox{if }m_2>0,\\
        m_1+1, & \mbox{if }m_2=0.
    \end{cases}
\]
Let $\nu(J_{m_1}\oplus uJ_{m_2})=m$, and $P\in M_{n\times m}(R)$ be a matrix such that
\[
J_{m_1}\oplus uJ_{m_2}=PP^T
\]
where $n=m_1+m_2$. Let $\mathcal{C}_P$ be a code generated by the rows of $P$ with type $\{a, b\}$. Define a map $\varphi:R^{m}\to R^n$ as
\[
\varphi(\mathbf{x})=\mathbf{x}P^T
\]
for $\mathbf{x}\in R^m$. Then,
\[
\varphi(\mathcal{C}_P)=\langle PP^T\rangle=\langle J_{m_1}\oplus uJ_{m_2}\rangle=R^{m_1}\oplus (uR)^{m_2}.
\]
Since the minimal number of generators for the image of an $R$-linear map is bounded above by that of its domain, we have
\[
m_1+m_2\le a+b\le m=\nu(J_{m_1}\oplus uJ_{m_2}).
\]
Thus, $\nu(J_{m_1}\oplus uJ_{m_2})=m_1+m_2$ if $m_2>0$.

Let $m_2=0$. Then we have
\[
m_1\le \nu(J_{m_1})\le m_1+1.
\]
Suppose that there is a matrix $Q\in M_{m_1\times m_1}(R)$ such that $J_{m_1}=QQ^T$. Since $J_{m_1}$ is invertible, $Q$ is also invertible. By reducing modulo $u$, we obtain an invertible matrix $\overline{Q}$ such that $\overline{J_{m_1}}=\overline{Q}\overline{Q}^T$. Since the diagonal entries of $J_{m_1}$ are zero, every row of $\overline{Q}$ is in
\[
H=\left\{\mathbf{x}\in \mathbb{F}_q^{m_1}~|~\sum_jx_j=0\right\}
\]
whose dimension is given as $\dim H=m_1-1$. Thus $\overline{Q}$ cannot generate $\mathbb{F}_q^{m_1}$, which is a contradiction since $Q$ is invertible. Hence $\nu(J_{m_1})=m_1+1$.
\end{proof}

\begin{theorem}\label{thm:exact_even}
Let $\mathcal C$ be a code of length $n$ over $R=\mathbb{F}_q+u\mathbb{F}_q$ with $q$ even. Let $G$ be a generator matrix of $\mathcal{C}$. Write $GG^T=A_0+uA_1$ for some matrices $A_0, A_1$ over $\mathbb{F}_q$. Let $r=\mathrm{rank}_{\mathbb{F}_q}(A_0)$, and let $\rho=\mathrm{rank}_{\mathbb{F}_q}(NA_1N^T)$, where the rows of $N$ form a basis of $\operatorname{Rad}(A_0)$. If
\begin{enumerate}
    \item[(i)] $\mathrm{Res}(\mathcal{C})$ is not self-orthogonal,
    \item[(ii)] for any $\mathbf{x}\in \mathrm{Res}(\mathcal{C})$, $\mathbf{x}\cdot \mathbf{x}=0$, and
    \item[(iii)] $\rho=0$,
\end{enumerate}
then
\[
n_s(\mathcal{C})=n+r+1.
\]
Otherwise,
\[
n_s(\mathcal{C})=n+r+\rho.
\]
\end{theorem}
\begin{proof}
    Here, we view $A_0$, $A_1$, $I_r$ and $J_r$ as matrices over $\mathbb{F}_q$ as all entries of those matrices are in $\mathbb{F}_q$. By Theorem~\ref{cong_2}, there exists an invertible matrix $P$ over $\mathbb{F}_q$ such that
    \[
    P A_0P^T=S\oplus\mathcal{O}
    \]
    where $S=I_r$ if $A_0$ has a nonzero diagonal entry, and $S=J_r$ if not. Since replacing $GG^T$ by $PGG^TP^T$ does not change $\nu(GG^T)$ by Proposition~\ref{embed-cong}, we may assume that $A_0=S\oplus \mathcal{O}$. Let $m$ denote the number of rows of $G$, and let
    \[
    A_1=\begin{pmatrix}
        X & Y\\Y^T & Z
    \end{pmatrix}
    \]
    where $X\in M_{r\times r}(\mathbb{F}_q)$, $Y\in M_{r\times (m-r)}(\mathbb{F}_q)$ and $Z\in M_{(m-r)\times (m-r)}(\mathbb{F}_q)$. Note that for $x=x_1+ux_2\in R$, we have $x^2=(x_1+ux_2)^2=x_1^2$. Thus, all diagonal entries of $A_1$ are zero. In particular, the diagonal entries of $X$ are all zero. 
    
    Let $V$ be a vector space of symmetric zero diagonal $r\times r$ matrices. Define a map
    \[
    \varphi:M_{r\times r}(\mathbb{F}_q)\to V
    \]
    as 
    \[
    \varphi(A)=AS+SA^T
    \]
    for $A\in M_{r\times r}(\mathbb{F}_q)$. Since $S$ is symmetric, for any $A\in M_{r\times r}(\mathbb{F}_q)$, $\varphi(A)$ is also symmetric. Moreover, since
    \[
    \varphi(A)_{ii}=(AS)_{ii}+(SA^T)_{ii}=(AS)_{ii}+(AS)_{ii}=0,
    \]
    its diagonal entries are all zero. Note that
    \[
    \ker \varphi=\{A\in M_{r\times r}(\mathbb{F}_q)~|~AS+SA^T=\mathcal{O}\}=\{A\in M_{r\times r}(\mathbb{F}_q)~|~AS=(AS)^T\}.
    \]
    Since the map $A\to AS$ is a vector space isomorphism on $M_{r\times r}(\mathbb{F}_q)$, the kernel $\ker \varphi$ has the same dimension as the space of symmetric $r\times r$ matrices, whose dimension is known as
    \[
    \frac{r(r+1)}{2}.
    \]
    Therefore,
    \[
    \dim(\operatorname{Im}\varphi)=r^2-\frac{r(r+1)}{2}=\frac{r(r-1)}{2}=\dim V.
    \]
    Thus, $\varphi$ is surjective. Then, there exists a matrix $Q_1\in M_{r\times r}(\mathbb{F}_q)$ such that
    \[
    Q_1S+SQ_1^T=X.
    \]
    Also, since $S$ is invertible, there is a matrix $Q_2\in M_{(m-r)\times r}(\mathbb{F}_q)$ such that $Y=SQ_2^T$.
    Define $U=I+uQ$ where
    \[
    Q=\begin{pmatrix}
        Q_1 & Q_3\\Q_2 & Q_4
    \end{pmatrix}
    \]
    for some $Q_3\in M_{r\times (m-r)}(\mathbb{F}_q)$ and $Q_4\in M_{(m-r)\times (m-r)}(\mathbb{F}_q)$. Now, we view $U$ as a matrix over $R$. Then we have
    \begin{align*}
    UGG^TU^T&=A_0+u(A_1+QA_0+A_0Q^T)\\
    &=A_0+uA_1+u\begin{pmatrix}
        Q_1S+SQ_1^T & SQ_2^T\\Q_2S & \mathcal{O}
    \end{pmatrix}\\
    &=A_0+\begin{pmatrix}
        \mathcal{O} & \mathcal{O}\\\mathcal{O} & uZ
    \end{pmatrix}\\&=S\oplus uZ.
    \end{align*}
    
    Note that the last $m-r$ rows of $P$ form a basis of $\operatorname{Rad}(A_0)$. Choose $N$ to be the matrix whose rows are precisely these last $m-r$ rows of $P$. Then $NA_1N^T$ is the lower right $(m-r)\times (m-r)$ submatrix of $PA_1P^T$, which is $Z$. Let $N'$ be a matrix whose rows form another basis of $\operatorname{Rad}(A_0)$. Then there is an invertible matrix $T\in GL_{m-r}(\mathbb{F}_q)$ such that $N'=TN$. Then we have
    \[
    N'A_1(N')^T=T(NA_1N^T)T^T.
    \]
    Since those are congruent, this gives
    \[
    \operatorname{rank}_{\mathbb{F}_q}(N'A_1(N')^T)=\operatorname{rank}_{\mathbb{F}_q}(NA_1N^T)=\operatorname{rank}_{\mathbb{F}_q}(Z).
    \]
    Note that $Z$ is a symmetric zero diagonal matrix, that is, an alternate matrix. Since $\mathrm{rank} Z=\rho$, $Z$ is congruent to $J_\rho\oplus \mathcal{O}$. Therefore, $GG^T$ is congruent to $S\oplus uJ_\rho\oplus \mathcal{O}$, and by Proposition~\ref{embed-cong}, $\nu(GG^T)=\nu(S\oplus uJ_\rho)$.

    Suppose that $\mathrm{Res}(\mathcal{C})$ is not self-orthogonal, but $\mathbf{x}\cdot\mathbf{x}=0$ for all $\mathbf{x}\in\mathrm{Res}(\mathcal{C})$. Then $A_0\ne \mathcal{O}$ and has zero diagonal. Thus,
    \[
    \nu(GG^T)=\nu(J_r\oplus u J_\rho)=\begin{cases}
        r+\rho, & \mbox{if }\rho>0,\\
        r+1, & \mbox{if }\rho=0
    \end{cases}
    \]
    by Lemma~\ref{nuzerodiag}. Otherwise,
    \[
    \nu(GG^T)=\nu(I_r\oplus u J_\rho)=r+\rho
    \]
    by Lemma~\ref{nunonzerodiag}. This completes the proof.
\end{proof}

\section{Shortest self-orthogonal embeddings of linear codes over {\it R} with odd characteristic}\label{sec4}
Throughout this section, we denote $R=\mathbb{F}_q+u\mathbb{F}_q$ for some $q=p^s$, $s\ge 1$ where $p$ is an odd prime.

For a linear code $\mathcal{C}$ over $R$ with generator matrix $G$, 
\[
\underbrace{\left[
G~|~
\cdots
~|~G
\right]}_{p\text{ copies}}.
\]
generates a self-orthogonal embedding of $\mathcal{C}$. Therefore, $\nu(GG^T)<\infty$ also when $R$ has odd characteristic.

For $m\ge 0$, define a matrix $K_m$ over $\mathbb{F}_q$ as
\[
K_m=\begin{pmatrix}
    \mathcal{O} & I_m\\I_m & \mathcal{O}
\end{pmatrix}.
\]
Let $A$ be any $m\times m$ symmetric matrix over $\mathbb{F}_q$. Take
\[
X=\left[I_m~|~\frac{1}{2}A\right].
\]
Then we have
\[
XK_mX^T=\frac{1}{2}A+\frac{1}{2}A^T=A.
\]
Thus, any symmetric matrix $A$ over $\mathbb{F}_q$ can be represented as $A=XK_mX^T$ for some matrix $X$ over $\mathbb{F}_q$ and for some $m\ge 1$. We define $\mu(A)$ as
\[
\mu(A)=\min\{m\ge 0~|~A=XK_mX^T\mbox{ for some matrix }X\mbox{ over }\mathbb{F}_q\}.
\]
It is straightforward that $\mu(A)=\mu(A\oplus\mathcal{O})$ for any symmetric matrix $A$ over $\mathbb{F}_q$.
\begin{remark}\label{mucong}
    Let $A$ and $A'$ be two congruent symmetric matrices over $\mathbb{F}_q$ with $q$ odd. Let $A'=PAP^T$ for some invertible matrix $P$ over $\mathbb{F}_q$. If $\mu(A)=m$, then there is a matrix $X$ over $\mathbb{F}_q$ such that $A=XK_mX^T$. Then we have
    \[
    A'=PAP^T=P(XK_mX^T)P^T=(PX)K_m(PX)^T.
    \]
    Thus, $\mu(A')\le m=\mu(A)$. Similarly, we have $\mu(A)\le \mu(A')$, i.e. $\mu(A)=\mu(A')$.
\end{remark}

\begin{lemma}\label{nuinv}
    Let $A$ be an invertible symmetric $k\times k$ matrix over $R$ such that all entries of $A$ are in $\mathbb{F}_q$ with $q$ odd. Then we have
    \[
    \nu(A)=\begin{cases}
        k, &\det(-A)\in (\mathbb{F}_q^\times)^2,\\
        k+1, & \det(-A)\notin (\mathbb{F}_q^\times)^2.
    \end{cases}
    \]
\end{lemma}
\begin{proof}
    Suppose that $\nu(A)=m$. Then there is a $k\times m$ matrix $B=B_0+uB_1$ over $R$ such that $A+BB^T=\mathcal{O}$. By reducing modulo $u$, we have $\overline{A}+\overline{B_0}\,\overline{B_0}^T=\mathcal{O}$ over $\mathbb{F}_q$. Conversely, if there is a matrix $B'$ over $\mathbb{F}_q$ such that $\overline{A}+B'(B')^T=\mathcal{O}$, then by viewing $B'$ as a matrix over $R$, we have $A+B'(B')^T=\mathcal{O}$ over $R$. This implies that $\nu(A)=\nu(\overline{A})$. Thus, we regard $A$ as a matrix over $\mathbb{F}_q$.

    Note that $\nu(A)\ge k$ since $A$ is invertible. We first show that $\nu(A)=k$ if and only if $\det(-A)\in (\mathbb{F}_q^\times)^2$. If $m=k$, then $\overline{B_0}$ must be an invertible matrix. Thus,
    \[
    \det (-A)=\det(\overline{B_0})^2\in (\mathbb{F}_q^\times)^2.
    \]
    Conversely, let $\det(-A)\in (\mathbb{F}_q^\times)^2$. Then by Theorem~\ref{congclass}, $-A$ is congruent to $I$. Hence, there is an invertible $k\times k$ matrix $P$ such that $-A=PP^T$. Then we have
    \[
    A+PP^T=\mathcal{O},
    \]
    i.e., $\nu(A)=k$.

    Second, we show that if $\det(-A)\notin (\mathbb{F}_q^\times)^2$, then $\nu(A)\le k+1$. Choose $a\in \mathbb{F}_q^\times$ such that $a\cdot\det(-A)\in (\mathbb{F}_q^\times)^2$. Then we have
    \[
    \det((-A)\oplus [a])\in (\mathbb{F}_q^\times)^2.
    \]
    Thus $(-A)\oplus [a]$ is congruent to $I_{k+1}$ and there is $P\in GL_{k+1}(\mathbb{F}_q)$ such that
    \[
    (-A)\oplus[a]=PP^T.
    \]
    Let $B_0$ be the matrix consisting of the first $k$ rows of $P$. Then we have $-A=B_0B_0^T$, which implies that $\nu(A)\le k+1$. This completes the proof.
\end{proof}

\begin{lemma}\label{remove-u-part}
Let $A$ be an $r\times r$ invertible symmetric matrix over $\mathbb{F}_q$ and $B$ be an $s\times s$ symmetric matrix over $\mathbb{F}_q$ with $q$ odd. Then, by viewing $A$ and $B$ as matrices over $R$,
\[
\nu(A\oplus uB)=\nu(A\oplus K_m)
\]
where $m=\mu(-B)$.
\end{lemma}
\begin{proof}
    We first show that
    \[
    \nu(A\oplus uB)\le \nu(A\oplus K_m).
    \]
    Let $\nu(A\oplus K_m)=t$. Then there is an $(r+2m)\times t$ matrix $D$ over $\mathbb{F}_q$ such that
    \[
    A\oplus K_m+DD^T=\mathcal{O}.
    \]
    Write the rows of $D$ as $d_1, \ldots, d_r, e_1, \ldots, e_m, f_1, \ldots, f_m$. Then we have
    \[
    d_id_j^T=-A_{ij},\quad d_ie_j^T=d_if_j^T=0, \quad e_ie_j^T=f_if_j^T=0 \quad\mbox{and} \quad e_if_j^T=-(K_m)_{ij}=-\delta_{ij}
    \]
    for every $i, j$. Since $m=\mu(-B)$, there exist $s\times m$ matrices $B_0, B_1$ over $\mathbb{F}_q$ such that
    \[
    -B=([B_0~|~B_1])K_m([B_0~|~B_1])^T=B_1B_0^T+B_0B_1^T.
    \]
    Define
    \[
    g_i=\sum_{k=1}^m(B_0)_{ik}e_k\quad\mbox{and}\quad h_i=-\sum_{k=1}^m(B_1)_{ik}f_k
    \]
    for $1\le i\le s$. Then we have
    \[
    g_ih_j^T+h_ig_j^T=(B_0B_1^T)_{ij}+(B_1B_0^T)_{ij}=-B_{ij}.
    \]
    Let $D_0$ be the $(r+s)\times t$ matrix over $\mathbb{F}_q$ with rows $d_1, \ldots, d_r, g_1, \ldots, g_s$ and $D_1$ be the $(r+s)\times t$ matrix over $\mathbb{F}_q$ with rows $0, \ldots, 0, h_1, \ldots, h_s$. Then we have
    \[
    D_0D_0^T=-A\oplus \mathcal{O}
    \]
    and
    \[
    D_0D_1^T+D_1D_0^T=\mathcal{O}\oplus(-B).
    \]
    Therefore,
    \[
    A\oplus uB+(D_0+uD_1)(D_0+uD_1)^T=\mathcal{O}.
    \]
    It follows that $\nu(A\oplus uB)\le t$.

    Next, we show that
    \[
    \nu(A\oplus K_m)\le \nu(A\oplus uB).
    \]
    Let $\nu(A\oplus uB)=t'$. Then, there is an $(r+s)\times t'$ matrix $E=E_0+uE_1$ where all entries of both $E_0$ and $E_1$ are in $\mathbb{F}_q$ such that
    \begin{equation}\label{eq1}
    A\oplus uB+EE^T=A\oplus uB+E_0E_0^T+u(E_0E_1^T+E_1E_0^T)=\mathcal{O}.
    \end{equation}
    Write the rows of $E_0$ as $e_1, \ldots, e_r, f_1, \ldots, f_s$ and the rows of $E_1$ as $g_1, \ldots, g_r, h_1, \ldots, h_s$. By reducing Equation~\ref{eq1} modulo $u$, we have
    \[
    e_ie_j^T=-A_{ij},\quad e_if_j^T=0\quad\mbox{and}\quad f_if_j^T=0
    \]
    for every $i, j$.
    Let $U=\operatorname{span}(e_1, \ldots, e_r)$. It is straightforward that $U$ is non-singular, and $f_i\in U^\perp$ for every $i$. By Equation~\ref{eq1}, we have
    \begin{equation}\label{eq2}
    f_ih_j^T+h_if_j^T=-B_{ij}
    \end{equation}
    for every $i, j$. Since $U$ is non-singular, we have $\mathbb{F}_q^{t'}=U\oplus U^\perp$. Since $f_i\in U^\perp$ for every $i$, replacing each $h_i$ by its $U^\perp$-component does not change Equation~\ref{eq2}. Thus, we may assume that $h_i\in U^\perp$ for every $i$. 
    
    Let $V=\operatorname{span}(f_1, \ldots, f_s)$ and $w$ be the Witt index of $U^\perp$. Then $V$ is a totally isotropic subspace of $U^\perp$. Let $V'$ be a maximal totally isotropic subspace of $U^\perp$ which contains $V$. Let $\alpha_1, \ldots, \alpha_w$ be a basis of $V'$. Then there is a $w$-dimensional totally isotropic subspace $W$ of $U^\perp$ with basis $\beta_1, \ldots, \beta_w$ such that
    \[
        \alpha_i\alpha_j^T=\beta_i\beta_j^T=0,\qquad \alpha_i\beta_j^T=\delta_{ij}
    \]
    by Theorem~\ref{Witt_basis}. Since $V'\oplus W$ is a non-singular subspace of $U^\perp$, we have
    \[
    U^\perp=(V'\oplus W)\perp (V'\oplus W)^\perp
    \]
    where $(V'\oplus W)^\perp$ is the orthogonal complement of $(V'\oplus W)$ in $U^\perp$. Write
    \[
    f_i=\sum_{k=1}^wG_{ik}\alpha_k\quad\mbox{and}\quad h_i=\sum_{k=1}^w(H_{ik}\beta_k+H'_{ik}\alpha_k)+g_i
    \]
    where $g_i\in (V'\oplus W)^\perp$. Note that $f_ig_j^T=0$ for every $i,\,j$.  Then we have
    \[
    -B_{ij}=f_ih_j^T+h_if_j^T=(GH^T)_{ij}+(HG^T)_{ij}.
    \]
    It follows that
    \[
    -B=GH^T+HG^T=[G~|~H]K_w[G~|~H]^T,
    \]
    which implies that $m\le w$. Let $E'$ be a matrix over $\mathbb{F}_q$ with rows
    \[
    e_1, \ldots, e_r, \alpha_1, \ldots, \alpha_m, -\beta_1, \ldots, -\beta_m.
    \]
    Then we have
    \[
    A\oplus K_m+E'(E')^T=A\oplus K_m+(-A\oplus -K_m)=\mathcal{O}.
    \]
    It follows that $\nu(A\oplus K_m)\le t'$. This completes the proof.
\end{proof}

\begin{theorem}\label{thm:exact_odd}
    Let $\mathcal{C}$ be a code of length $n$ over $R=\mathbb{F}_q+u\mathbb{F}_q$ with $q$ odd. Let $G$ be a generator matrix of $\mathcal{C}$. Write $GG^T=A_0+uA_1$ for some matrices $A_0, A_1$ over $\mathbb{F}_q$. Let $r=\operatorname{rank}_{\mathbb{F}_q}(A_0)$, and let $P$ be an invertible matrix over $\mathbb{F}_q$ such that
    \[
    PA_0P^T=S\oplus \mathcal{O}
    \]
    where $S=I_r$ or $S=I_{r-1}\oplus [z]$ for some $z\in \mathbb{F}_q^\times\backslash(\mathbb{F}_q^\times)^2$. Write
    \[
    PA_1P^T=\begin{pmatrix}
        X & Y\\Y^T & Z
    \end{pmatrix}
    \]
    where $X$ is a $r\times r$ symmetric matrix over $\mathbb{F}_q$. Let $m=\mu(-Z)$. Then
    \[
    n_s(\mathcal{C})=\begin{cases}
        n+r+2m, & (-1)^{r+m}\det(S)\in (\mathbb{F}_q^\times)^2,\\
        n+r+2m+1, & (-1)^{r+m}\det(S)\notin (\mathbb{F}_q^\times)^2.
    \end{cases}
    \]
\end{theorem}
\begin{proof}
First, we note that $m$ is independent to the choice of $P$ by Remark~\ref{mucong}. Since replacing $GG^T$ by $PGG^TP^T$ does not change $\nu(GG^T)$ by Proposition~\ref{embed-cong}, we may assume that $A_0=S\oplus \mathcal{O}$. Let $k$ denote the number of rows of $G$. Let
\[
Q_1=-\frac{XS^{-1}}{2}\quad\mbox{and}\quad Q_2=-Y^TS^{-1}.
\]
Define a $k\times k$ matrix $Q$ over $\mathbb{F}_q$ as
\[
Q=\begin{pmatrix}
    Q_1 & Q_3\\Q_2 & Q_4
\end{pmatrix}
\]
for some matrices $Q_3$ and $Q_4$. Let $U=I+uQ$. Then we have
\begin{align*}
U(A_0+uA_1)U^T&=A_0+u(A_1+QA_0+A_0Q^T)\\
&=A_0+uA_1+u\begin{pmatrix}
        Q_1S+SQ_1^T & SQ_2^T\\Q_2S & \mathcal{O}
    \end{pmatrix}\\
    &=A_0+\begin{pmatrix}
        \mathcal{O} & \mathcal{O}\\\mathcal{O} & uZ
    \end{pmatrix}\\&=S\oplus uZ.
\end{align*}

Then, by Lemma~\ref{remove-u-part}, we have
\[
\nu(GG^T)=\nu(S\oplus uZ)=\nu(S\oplus K_m).
\]
Note that $S\oplus K_m$ is a nonsingular symmetric matrix over $\mathbb{F}_q$ of size $r+2m$. Moreover we have
\begin{align*}
\det(-(S\oplus K_m))&=(-1)^{r+2m}\det (S)\det(K_m)\\&=(-1)^{r+m}\det(S)
\end{align*}
since $\det(K_m)=(-1)^m$. Therefore, by Lemma~\ref{nuinv}, 
\[
\nu(GG^T)=\nu(S\oplus K_m)=\begin{cases}
        r+2m, &(-1)^{r+m}\det(S)\in (\mathbb{F}_q^\times)^2,\\
        r+2m+1, & (-1)^{r+m}\det(S)\notin (\mathbb{F}_q^\times)^2.
    \end{cases}
\]
This completes the proof.
\end{proof}

\section{Construction of shortest self-orthogonal embeddings of linear codes over {\it R}}

The following theorem shows that every self-orthogonal code over $R$ with nonzero free rank can be viewed as a shortest self-orthogonal embedding of another code.
\begin{theorem}\label{SOcofromsmall}
    Let $\mathcal{C}$ be a self-orthogonal code of length $n$ over $R=\mathbb{F}_q+u\mathbb{F}_q$ for some prime power $q$. Let $\mathcal{C}$ is of type $\{k_1,\,k_2\}$ where $k_1>0$. For every integer $n_0$ satisfying $n-k_1\le n_0<n$, there is a code $\mathcal{C}_0$ over $R$ with length $n_0$ whose shortest self-orthogonal embedding is $\mathcal{C}$.
\end{theorem}
\begin{proof}
    Let $G$ be the standard generator matrix of $\mathcal{C}$ given as
    \[
    G=\begin{pmatrix}
        I_{k_1} & A & B_1+uB_2\\
        \mathcal{O} & uI_{k_2} & uD
    \end{pmatrix}
    \]
    where $A,\,B_1,\,B_2$ and $D$ are matrices over $\mathbb{F}_q$. Let $t=n-n_0$, and choose $t$ columns in the first identity block of $G$. After coordinate permutation, we write
    \[
    G=[G_0~|~B]
    \]
    where $B$ consists of these $t$ columns. Let $\mathcal{C}_0$ be the code generated by $G_0$. Define an $R$-linear map $\pi:\mathcal{C}\to \mathcal{C}_0$ as the puncturing of these $t$ coordinates. For $\mathbf{c}\in\mathcal{C}$, let $\pi(\mathbf{c})=0$, that is, $\mathbf{c}=(\mathbf{0}, \mathbf{c}_0)$ for some $\mathbf{c}_0\in R^t$. Since the columns of $B$ are distinct standard basis vectors, we have $B^TB=I$, and it follows that
    \[
    \mathbf{c}_0=\mathbf{c}_0B^TB=(\mathbf{c}G^T)B=\mathbf{0}.
    \]
    This implies that $\pi$ is an isomorphism, and $\mathcal{C}_0$ is a code of type $\{k_1,\,k_2\}$.

    We first show that $G_0$ is a generator matrix of $\mathcal{C}_0$. Define two maps $\Phi_G:R^{k_1}\times \mathbb{F}_q^{k_2}\to \mathcal{C}$ and $\Phi_{G_0}:R^{k_1}\times \mathbb{F}_q^{k_2}\to \mathcal{C}_0$, respectively as follows:
    \[
    \Phi_G(\mathbf{x}, \mathbf{y})=(\mathbf{x}, \mathbf{y})G\quad\mbox{and}\quad
    \Phi_{G_0}(\mathbf{x}, \mathbf{y})=(\mathbf{x}, \mathbf{y})G_0.
    \]
    Clearly, $\Phi_G$ is a bijection, and $\Phi_{G_0}=\pi \circ\Phi_G$ is also a bijection. Moreover, the last $k_2$ rows of $G_0$ lie in $(uR)^{n_0}$, while its first $k_1$ rows cannot lie in $(uR)^{n_0}$, since otherwise multiplying it by $u$ would yield zero, contradicting the injectivity of $\Phi_{G_0}$. Therefore, $G_0$ is a generator matrix of $\mathcal C_0$, and it follows that $\mathcal{C}$ is a self-orthogonal embedding of $\mathcal{C}_0$.

    Next, we prove the shortestness. Note that
    \[
    \mathcal{O}=GG^T=G_0G_0^T+BB^T.
    \]
    Thus $\operatorname{rank}_{\mathbb{F}_q}(\overline{G_0G_0^T})=\operatorname{rank}_{\mathbb{F}_q}(-\overline{BB^T})=t$. Suppose that $[G_0~|~X]$ generates a self-orthogonal code. Then we have
    \[
    G_0G_0^T+XX^T=\mathcal{O}
    \]
    for some $(k_1+k_2)\times s$ matrix $X$. Then we have
    \[
    t=\operatorname{rank}_{\mathbb{F}_q}(-\overline{G_0G_0^T})=\operatorname{rank}_{\mathbb{F}_q}(\overline{XX^T})\le \operatorname{rank}_{\mathbb{F}_q}(\overline{X})\le s.
    \]
    Thus, $\mathcal{C}$ is a shortest self-orthogonal embedding of $\mathcal{C}_0$.
\end{proof}

\begin{remark}
    In Theorem~\ref{SOcofromsmall}, we assumed that the free rank $k_1$ of a code $\mathcal{C}$ satisfies $k_1>0$. If $k_1=0$, then this means that any element of a generator matrix of $\mathcal{C}$ is a multiple of $u$. Then every punctured code of $\mathcal{C}$ is also self-orthogonal since a code of free rank zero is self-orthogonal. Consequently, any self-orthogonal code of free rank zero can be obtained as a self-orthogonal embedding of another self-orthogonal code of free rank zero, but not the shortest one. Therefore, we do not consider this case in detail.
\end{remark}

Let $\mathcal{C}$ be a linear code over $R=\mathbb{F}_q+u\mathbb{F}_q$ with generator matrix $G$. In Theorem~\ref{thm:exact_even}, for even $q$, we constructed a congruence between $GG^T$ and a matrix of the form $S\oplus uJ_\rho\oplus \mathcal{O}$, where $S$ is either $I$ or $J$. In Lemma~\ref{nunonzerodiag} and Lemma~\ref{nuzerodiag}, we explicitly constructed a matrix $B$ having the minimum possible number of columns such that
\[
S\oplus uJ_\rho+ BB^T=\mathcal{O}.
\]

In Theorem~\ref{thm:exact_odd}, for odd $q$, we also constructed a congruence between $GG^T$ and a matrix of the form $S\oplus uZ$, where $Z$ is as defined in Theorem~\ref{thm:exact_odd}. In Lemma~\ref{nuinv} and Lemma~\ref{remove-u-part}, we explicitly constructed a matrix $B$ with the minimum number of columns such that
\[
S\oplus uZ+BB^T=\mathcal{O}.
\]

Consequently, these constructions yield an explicit construction of a shortest self-orthogonal embedding of $\mathcal{C}$. In this section, we present an algorithm for obtaining all other shortest self-orthogonal embeddings of $\mathcal{C}$ from this particular one.

Let $\mathcal{C}$ be a linear code over $R=\mathbb{F}_q+u\mathbb{F}_q$ with generator matrix $G$. Let $GG^T=A_0+uA_1$ where $A_0$ and $A_1$ are matrices over $\mathbb{F}_q$. Let $B$ be a $k\times m$ matrix over $R$ where $k$ is the number of rows of $G$ and $m=\nu(GG^T)$, and let $B=X+uY$ where $X$ and $Y$ are matrices over $\mathbb{F}_q$. We assume that $m>0$. Then the matrix $[G~|~B]$ generates a shortest self-orthogonal embedding of $\mathcal{C}$ if and only if
\begin{align*}
    \mathcal{O}=[G~|~B][G~|~B]^T=(A_0+XX^T)+u(A_1+XY^T+YX^T).
\end{align*}
Therefore, finding a shortest self-orthogonal embedding of $\mathcal{C}$ is equivalent to finding a solution $(X,Y)$ to the following system:
\[
XX^T=-A_0\quad\mbox{and}\quad XY^T+YX^T=-A_1.
\]
For a matrix $X\in M_{k\times m}(\mathbb{F}_q)$, let $\varphi_X:M_{k\times m}(\mathbb{F}_q)\to \operatorname{Sym}_k(\mathbb{F}_q)$ be the linear map given as 
\[
\varphi_X(Y)=XY^T+YX^T.
\]
We define
\[
\chi_{m, G}=\{X\in M_{k\times m}(\mathbb{F}_q)~|~XX^T=-A_0, \mbox{ and }\varphi_X^{-1}(-A_1)\ne \emptyset\}.
\]
Throughout this section, we fix the code $\mathcal{C}$ and retain the notation introduced in Section~\ref{sec3} and Section~\ref{sec4}.

\begin{proposition}\label{liftcriterion}
    For a matrix $A\in M_{k\times m}(\mathbb{F}_q)$ for some prime power $q$, let $N_A$ be a matrix whose rows form a basis of $\ker A$. Let $M$ be a $k\times k$ symmetric matrix over $\mathbb{F}_q$, and when $q$ is even, assume further that $M$ is alternate. Then there is a matrix $B$ such that $\varphi_A(B)=M$ if and only if $N_AMN_A^T=\mathcal{O}$.
\end{proposition}
\begin{proof}
    Assume that such a matrix $B$ exists. Then we obtain
    \[
    N_AMN_A^T=N_A(\varphi_A(B))N_A^T=N_A(AB^T+BA^T)N_A^T=\mathcal{O}.
    \]

    Conversely, assume that $N_AMN_A^T=\mathcal{O}$. Let $\operatorname{rank}(A)=a$, and let $U\in GL_k(\mathbb{F}_q)$ be a matrix whose last $k-a$ rows form a basis of $\ker A$. Then we have
    \[
    UMU^T=\begin{pmatrix}
        M_1 & M_2\\M_2^T & \mathcal{O}
    \end{pmatrix}
    \]
    for some matrices $M_1$ and $M_2$. Choose $T\in M_a(\mathbb{F}_q)$ such that $T+T^T=M_1$. Choose a matrix $V\in GL_m(\mathbb{F}_q)$ such that
    \[
    UAV=\begin{pmatrix}
    I_a & \mathcal{O}\\\mathcal{O} & \mathcal{O}
    \end{pmatrix}.
    \]
    Take
    \[
    B=U^{-1}\begin{pmatrix}
        T & \mathcal{O}\\M_2^T & \mathcal{O} 
    \end{pmatrix}V^T.
    \]
    Then we have
    \begin{align*}
    U(\varphi_A(B))U^T&=U(AB^T+BA^T)U^T\\&=\begin{pmatrix}
    I_a & \mathcal{O}\\\mathcal{O} & \mathcal{O}
    \end{pmatrix}\begin{pmatrix}
        T & \mathcal{O}\\M_2^T & \mathcal{O} 
    \end{pmatrix}^T+\begin{pmatrix}
        T & \mathcal{O}\\M_2^T & \mathcal{O} 
    \end{pmatrix}\begin{pmatrix}
    I_a & \mathcal{O}\\\mathcal{O} & \mathcal{O}
    \end{pmatrix}^T\\&=UMU^T
    \end{align*}
    This completes the proof.
\end{proof}

\begin{lemma}\label{X-rank}
    Let $\mathcal{C}$ be a code over $R=\mathbb{F}_q+u\mathbb{F}_q$ for some prime power $q$ with generator matrix $G$. Write $GG^T=A_0+uA_1$ for some matrices $A_0,\,A_1$ over $\mathbb{F}_q$. Let $N$ be a matrix whose rows form a basis of $\operatorname{Rad}(A_0)$, and let $\tilde{N}=NA_1N^T$. For $X\in \chi_{m, G}$, the rank of $X$ is given as
    \[
    \operatorname{rank}(X)=\begin{cases}
        \operatorname{rank}(A_0)+\mu(-\tilde{N}), & \mbox{if }q\mbox{ is odd},\\
        \operatorname{rank}(A_0)+\operatorname{rank}(\tilde{N})/2, & \mbox{if }q\mbox{ is even}.
    \end{cases}
    \]
\end{lemma}
\begin{proof}
Let $\mathcal{C}_X$ be a code generated by $X$ and let $\ell_X=\dim\operatorname{Hull}(\mathcal{C}_X)$. For $\mathbf{u}\in \operatorname{Rad}(A_0)$ and for any $\mathbf{v}\in \mathbb{F}_q^k$, we have
\[
(\mathbf{u}X)(\mathbf{v}X)^T=-\mathbf{u}A_0\mathbf{v}^T=0.
\]
Thus, $\mathbf{u}X\in \operatorname{Hull}(\mathcal{C}_X)$. Conversely, if $\mathbf{u}'X\in\operatorname{Hull}(\mathcal{C}_X)$, then $\mathbf{u}'XX^T=-\mathbf{u}'A_0=\mathbf{0}$. Thus, $\mathbf{u}'\in \operatorname{Rad}(A_0)$. It follows that
\[
\operatorname{Hull}(\mathcal{C}_X)=\langle NX\rangle.
\]
Since $\ker X=\{\mathbf{u}\in \mathbb{F}_q^k~|~\mathbf{u}X=\mathbf{0}\}\subseteq \operatorname{Rad}(A_0)$, we have
\begin{align*}
\ell_X=\operatorname{rank}(NX)&=\dim \operatorname{Rad}(A_0)-\dim \ker X\\&=(k-\operatorname{rank}(A_0))-(k-\operatorname{rank}(X))\\&=\operatorname{rank}(X)-\operatorname{rank}(A_0).
\end{align*}
Note that
\[
\ell_X\le \dim\mathcal{C}_X^\perp=m-\dim\mathcal{C}_X=m-(\ell_X+\operatorname{rank}(A_0)).
\]
Then we have
\begin{equation}
\ell_X\le (m-\operatorname{rank}(A_0))/2.
\end{equation}
Choose a matrix $Y$ such that $\varphi_X(Y)=-A_1$. Let $X'=NX$ and $Y'=NY$. This gives
\[
-\tilde{N}=-NA_1N^T=N(XY^T+YX^T)N^T=(X'(Y')^T+Y'(X')^T)
\]

We first consider the case that $q$ is odd. Let $X_0$ be the matrix whose columns form a basis of the column space of $X'$ whose dimension is $\ell_X$. Then we can represent $X'$ as $X'=X_0E$ for some matrix  $E$ with row rank $\ell_X$. Define $Y_0=Y'E^T$. Then we have
\[
-\tilde{N}=X_0Y_0^T+Y_0X_0^T=[X_0~|~Y_0](K_{\ell_X})[X_0~|~Y_0]^T,
\]
which implies that $\mu(-\tilde{N})\le \ell_X$. Since
\[
\ell_X\le \left\lfloor\frac{m-\operatorname{rank}(A_0)}{2}\right\rfloor\le \mu(-\tilde{N})
\]
by Theorem~\ref{thm:exact_odd}, it follows that $\ell_X=\mu(-\tilde{N})$.

Next, suppose that $q$ is even. Then we have
\[
\operatorname{rank}(\tilde{N})\le \operatorname{rank}(X'(Y')^T)+\operatorname{rank}(Y'(X')^T)\le 2\ell_X.
\]
Note that $2\ell_X\le m-\operatorname{rank}(A_0)$, which is equal to $\operatorname{rank}(\tilde{N})$ or $1$ by Theorem~\ref{thm:exact_even}. If it is equal to $1$, then $\operatorname{rank}(\tilde{N})$ must be zero. Therefore, in either case, we have
\[
\ell_X=\frac{\operatorname{rank}(\tilde{N})}{2}.
\]
\end{proof}

Suppose that $q$ is even. In Theorem~\ref{thm:exact_even}, we showed that $n_s(\mathcal{C})=n+r+1$ for a code $\mathcal{C}$ of length $n$ over $\mathbb{F}_q+u\mathbb{F}_q$ if the following three conditions hold:
\begin{enumerate}
    \item[(i)] $\mathrm{Res}(\mathcal{C})$ is not self-orthogonal,
    \item[(ii)] for any $\mathbf{x}\in \mathrm{Res}(\mathcal{C})$, $\mathbf{x}\cdot \mathbf{x}=0$, and
    \item[(iii)] $\rho=0$,
\end{enumerate}
where $r$ and $\rho$ are as defined in Theorem~\ref{thm:exact_even}. We refer to a code satisfying these three conditions as a code of \textit{Case I}. Otherwise, we refer to it as a code of \textit{Case II}.

\begin{lemma}\label{allonevector}
    Let $\mathcal{C}$ be a code over $R=\mathbb{F}_q+u\mathbb{F}_q$ with $q$ even, and let $G$ be a generator matrix of $\mathcal{C}$. Let $\mathbf{1}_m$ be the all-one vector of length $m$ for some $m>0$. If $\mathcal{C}$ is of Case I, then $\mathbf{1}_m\notin \langle X\rangle$ for every $X\in \chi_{m, G}$. Otherwise, $\mathbf{1}_m\in \langle X\rangle$ for every $X\in \chi_{m, G}$.
\end{lemma}
\begin{proof}
    Let $\mathcal{C}_X=\langle X\rangle$. By Lemma~\ref{X-rank}, we have
    \[
    \dim \mathcal{C}_X=\operatorname{rank}(A_0)+\frac{\operatorname{rank}(\tilde{N})}{2}
    \]
    where $\tilde{N}=NA_1N^T$ and $N$ is a matrix whose rows form a basis of $W=\operatorname{Rad}(A_0)$. Note that 
    \[
    \dim\operatorname{Hull}(\mathcal{C}_X)=\dim\mathcal{C}_X-\operatorname{rank}(A_0)=\frac{\operatorname{rank}(\tilde{N})}{2}.
    \]
    Assume that $\mathcal{C}$ is of Case II. Since $\dim \mathcal{C}_X^\perp=m-\dim\mathcal{C}_X=\frac{\operatorname{rank}(\tilde{N})}{2}$ and $\operatorname{Hull}(\mathcal{C}_X)\subseteq \mathcal{C}_X^\perp$, we have $\operatorname{Hull}(\mathcal{C}_X)= \mathcal{C}_X^\perp$. For any $\mathbf{u}\in \operatorname{Hull}(\mathcal{C}_X)$, we have
    \[
    (\mathbf{u}\cdot \mathbf{1}_m)^2= \mathbf{u}\cdot\mathbf{u}=0,
    \]
    which implies that $\mathbf{u}\cdot \mathbf{1}_m=0$. Therefore, $\mathbf{1}_m\in \operatorname{Hull}(\mathcal{C}_X)^\perp=\mathcal{C}_X$.

    Conversely, assume that $\mathcal{C}$ is of Case I. Then $A_0\ne \mathcal{O}$ is alternate, $\operatorname{rank}(\tilde{N})=0$ and $m=\operatorname{rank}(A_0)+1$. Then, by Lemma~\ref{X-rank}, $\dim\mathcal{C}_X=\operatorname{rank}(A_0)$ and $\operatorname{Hull}(\mathcal{C}_X)=\{0\}$. Since $A_0$ is alternate, for any $\mathbf{u}=\mathbf{u}_0X\in \mathcal{C}_X$ for some $\mathbf{u}_0\in \mathbb{F}_q^k$, we have
    \[
   (\mathbf{u}\cdot \mathbf{1}_m)^2= \mathbf{u}\cdot\mathbf{u}=-\mathbf{u}_0A_0\mathbf{u}_0^T=0,
    \]
    which implies that $\mathbf{u}\cdot \mathbf{1}_m=0$. Thus, $\mathbf{1}_m\in \mathcal{C}_X^\perp$. Since $\operatorname{Hull}(\mathcal{C}_X)=\{0\}$, $\mathbf{1}_m$ cannot be in $\mathcal{C}_X$. This completes the proof.
\end{proof}

For $P\in GL_k(\mathbb{F}_q)$, let
    \[
    K_P=(\ker X_*)P^{-1}=\{\mathbf{x}P^{-1}~|~\mathbf{x}\in \ker X_*\}.
    \]
Define
    \[
    \Gamma_X=\{P\in GL_k(\mathbb{F}_q)~|~PA_0P^T=A_0~\mbox{and}~\mathbf{u}A_1\mathbf{v}^T=0~\mbox{for every}~\mathbf{u},\, \mathbf{v}\in K_P\}.
    \]

\begin{theorem}\label{everychi}
    Let $\mathcal{C}$ be a code over $R=\mathbb{F}_q+u\mathbb{F}_q$ for some prime power $q$ with generator matrix $G$. Let $X_*\in \chi_{m, G}$. Then
    \[
    \chi_{m, G}=\{PX_*Q~|~P\in\Gamma_X, Q\in O_m(\mathbb{F}_q)\}
    \]
    where $O_m(\mathbb{F}_q)$ is the orthogonal group of degree $m$ over $\mathbb{F}_q$.
\end{theorem}
\begin{proof}
    For $P\in\Gamma_X$ and $Q\in O_m(\mathbb{F}_q)$, we have
    \[
    (PX_*Q)(PX_*Q)^T=PX_*QQ^TX_*^TP^T=-A_0
    \]
    and
    \[
    \ker PX_*Q=\ker PX_*=\{\mathbf{u}\in \mathbb{F}_q^k~|~\mathbf{u}PX_*=\mathbf{0}\}=\{\mathbf{v}P^{-1}~|~\mathbf{v}\in \ker X_*\}=(\ker X_*)P^{-1}=K_P.
    \]
    Let $N$ be the matrix whose rows form a basis of $\ker PX_*Q=K_P$. By the definition of $\Gamma_X$, we have $NA_1N^T=\mathcal{O}$. Then, by Proposition~\ref{liftcriterion}, there is a matrix $Y$ such that $\varphi_{PX_*Q}(Y)=-A_1$. Therefore, $PX_*Q\in \chi_{m, G}$.
    
    Next, we consider the opposite direction. Take any $X\in\chi_{m, G}$. Since $\operatorname{rank}(X)=\operatorname{rank}(X_*)$ by Lemma~\ref{X-rank}, we have $\dim\ker X=\dim\ker X_*$. Let $V$ be a complement of $\operatorname{Rad}(A_0)$ in $\mathbb{F}_q^k$. Then $\mathbb{F}_q^k=V\perp \operatorname{Rad}(A_0)$ and the restriction of $A_0$ to $V$ is nondegenerate. Since $\ker X, \ker X_*\subseteq \operatorname{Rad}(A_0)$ have the same dimension, there is an invertible matrix $P\in GL_k(\mathbb{F}_q)$ such that
    \[
    \mathbf{v}P=\mathbf{v}\quad\mbox{for all }\mathbf{v}\in V,\quad (\operatorname{Rad}(A_0))P=\operatorname{Rad}(A_0),\quad\mbox{and}\quad(\ker X)P=\ker X_*.
    \]
    Note that $K_P=(\ker X_*)P^{-1}=\ker X$. Choose $Y$ such that $\varphi_X(Y)=-A_1$. Then for $\mathbf{u},\,\mathbf{v}\in \ker X$, we have
    \begin{equation}\label{everychi-eq}
        \mathbf{u}A_1\mathbf{v}^T=-\mathbf{u}(XY^T+YX^T)\mathbf{v}^T=0.
    \end{equation}
    For any two vectors $\mathbf{x}, \mathbf{y}\in \mathbb{F}_q^k$, decompose them into $\mathbf{x}=\mathbf{x}_0+\mathbf{x}_1$ and $\mathbf{y}=\mathbf{y}_0+\mathbf{y}_1$ where $\mathbf{x}_0, \mathbf{y}_0\in V$ and $\mathbf{x}_1, \mathbf{y}_1\in \operatorname{Rad}(A_0)$. Then we have
    \begin{align*}
    (\mathbf{x}P)A_0(\mathbf{y}P)^T=(\mathbf{x}_0+\mathbf{x}_1P)A_0(\mathbf{y}_0+\mathbf{y}_1P)^T=\mathbf{x}_0A_0\mathbf{y}_0^T=\mathbf{x}A_0\mathbf{y}^T.
    \end{align*}
    This gives $PA_0P^T=A_0$. Therefore $P\in\Gamma_X$, and the map $\xi:\langle PX_*\rangle\to \langle X\rangle$ defined as
    \[
    \xi(\mathbf{u}PX_*)=\mathbf{u}X
    \]
    is an isometry. Next, we divide the proof into three cases. 

    First, let $q$ be odd. By Theorem~\ref{witt_e}, $\xi$ extends to an isometry $\tilde{\xi}$ on $\mathbb{F}_q^m$. Note that there is an orthogonal matrix $Q\in O_m(\mathbb{F}_q)$ such that $\tilde{\xi}(\mathbf{z})=\mathbf{z}Q$ for every $\mathbf{z}\in\mathbb{F}_q^m$. Then for any $\mathbf{a}\in\mathbb{F}_q^k$, we have
    \[
    \mathbf{a}PX_*Q=\tilde{\xi}(\mathbf{a}PX_*)=\xi(\mathbf{a}PX_*)=\mathbf{a}X,
    \]
    which gives $PX_*Q=X$.

    Second, let $q$ be even and $\mathcal{C}$ be a code of Case I. By Lemma~\ref{allonevector}, all-one vector $\mathbf{1}_m$ is contained in neither $\langle PX_*\rangle$ nor $\langle X\rangle$. Again, by Theorem~\ref{witt_e}, $\xi$ extends to an isometry $\tilde{\xi}$ on $\mathbb{F}_q^m$, which implies that $PX_*Q=X$.

    Finally, let $q$ be even and $\mathcal{C}$ be a code of Case II. By Lemma~\ref{allonevector}, all-one vector $\mathbf{1}_m$ is contained in both $\langle X_*\rangle$ and $\langle X\rangle$. Choose $\mathbf{c}, \mathbf{c}_*\in \mathbb{F}_q^k$ such that $\mathbf{c}X=\mathbf{c}_*X_*=\mathbf{1}_m$. Then for any $\mathbf{a}\in\mathbb{F}_q^k$, we have
    \begin{align*}
    (\mathbf{a}X\cdot \mathbf{c}X)^2&=
    (\mathbf{a}X\cdot \mathbf{1}_m)^2\\&=
    (\mathbf{a}X\cdot\mathbf{a}X)\\&=
    -\mathbf{a}A_0\mathbf{a}^T\\&=(\mathbf{a}X_*\cdot\mathbf{a}X_*)\\&=
    (\mathbf{a}X_*\cdot \mathbf{1}_m)^2\\&=
    (\mathbf{a}X_*\cdot \mathbf{c}_*X_*)^2.
    \end{align*}
    This gives that $\mathbf{a}X\cdot \mathbf{c}X=\mathbf{a}X_*\cdot \mathbf{c}_*X_*$ which implies
    \[
    -\mathbf{a}A_0\mathbf{c}^T=
    \mathbf{a}X\cdot \mathbf{c}X=
    \mathbf{a}X_*\cdot \mathbf{c}_*X_*=
    -\mathbf{a}A_0\mathbf{c}_*^T.
    \]
    This gives $\mathbf{c}-\mathbf{c}_*\in \operatorname{Rad}(A_0)$. Choose a basis of $\mathbb{F}_q^k=V\oplus\operatorname{Rad}(A_0)$ such that $A_0=A_0'\oplus \mathcal{O}$ for some matrix $A_0'$. Then, we can write $\mathbf{c}=(\mathbf{c}_0, \mathbf{z})$ and $\mathbf{c}_*=(\mathbf{c}_0, \mathbf{z}_*)$ for some $\mathbf{c}_0\in V$ and $\mathbf{z}, \mathbf{z}_*\in \operatorname{Rad}(A_0)$. Note that every matrix
    \[
    P'=\begin{pmatrix}
        I & B\\\mathcal{O} & D
    \end{pmatrix}
    \]
    for some invertible matrix $D$ satisfies $P'A_0(P')^T=A_0$. If $\mathbf{c}_0\ne 0$, then choose $D$ satisfying $(\ker X)D=\ker X_*$ and choose $B$ such that $\mathbf{c}_0B=\mathbf{z}_*-\mathbf{z}D$. Otherwise, if $\mathbf{c}_0=0$, then take $B=\mathcal{O}$ and choose $D$ satisfying $(\ker X)D=\ker X_*$ and $\mathbf{z}D=\mathbf{z}_*$. Then, in both cases, $P'\in \Gamma_X$ and we have
    \[
    \mathbf{c}P'-\mathbf{c}_*=(\mathbf{c}_0, \mathbf{c}_0B+\mathbf{z}D)-(\mathbf{c}_0, \mathbf{z}_*)=0.
    \]
    Define an isometry $\xi':\langle P'X_*\rangle\to \langle X\rangle$ as $\xi'(\mathbf{u}P'X_*)=\mathbf{u}X$ for every $\mathbf{u}\in \mathbb{F}_q^k$. Then we have
    \[
    \xi'(\mathbf{1}_m)=\xi'((\mathbf{c}_*(P')^{-1})P'X_*)=\mathbf{c}_*(P')^{-1}X=\mathbf{c}X=\mathbf{1}_m.
    \]
    Again, by Theorem~\ref{witt_e}, $\xi'$ extends to an isometry $\tilde{\xi}'$ on $\mathbb{F}_q^m$, which implies that $P'X_*Q=X$. Moreover, an argument similar to that used in~\eqref{everychi-eq} shows that $P'\in\Gamma_X$. This completes the proof.
\end{proof}

\begin{remark}\label{conrem}
    For $P\in \Gamma_X$, note that $\ker PX_*=(\ker X_*)P^{-1}$ as shown in Theorem~\ref{everychi}. Let $N_{PX_*}$ be the matrix whose rows form a basis of $\ker PX_*$. Since $(\ker PX_*)P=\ker X_*$, the rows of $N_{PX_*}P$ form a basis of $\ker X_*$. Thus, by choosing a suitable basis of $\ker X_*$, we have
    \[
    N_{PX_*}P=N_{X_*}
    \]
    where $N_{X_*}$ is the matrix whose rows form a basis of $\ker X_*$. Then by the definition of $\Gamma_X$ and by Proposition~\ref{liftcriterion}, we have
    \begin{align*}
    N_{PX_*}(PA_1P^T-A_1)N_{PX_*}^T&=N_{PX_*}(PA_1P^T)N_{PX_*}^T-N_{PX_*}A_1N_{PX_*}^T\\&=N_{X_*}A_1N_{X_*}^T-N_{PX_*}A_1N_{PX_*}^T\\&=\mathcal{O}.
    \end{align*}
    Therefore, by Proposition~\ref{liftcriterion}, there is a matrix $S$ such that
    \[
    \varphi_{PX_*}(S)=PA_1P^T-A_1.
    \]
    For each $P\in\Gamma_X$, we fix one such matrix and denote it by $S_P$.
\end{remark}

\begin{theorem}\label{optimalconstruct}
    Let $\mathcal{C}$ be a code over $\mathbb{F}_q+u\mathbb{F}_q$ where $q$ is a prime power. Let $G$ be a generator matrix of $\mathcal{C}$. Suppose that the matrix
    \[
    \tilde{G}=[G~|~B_*]
    \]
    generates a shortest self-orthogonal embedding $\tilde{\mathcal{C}}$ of $\mathcal{C}$, where $B_*=X_*+uY_*$ for some $X_*, Y_*$ over $\mathbb{F}_q$. Then, a code $\mathcal{C}'$ is a shortest self-orthogonal embedding of $\mathcal{C}$ if and only if it has a generator matrix
    \[
    G'=[G~|~(PX_*+u(PY_*+S_P+PH))Q]
    \]
    where $P\in \Gamma_X$, $H\in \ker\varphi_{X_*}$ and $Q\in O_m(\mathbb{F}_q)$.
\end{theorem}
\begin{proof}
    Let $X'=PX_*$ and $Y'=PY_*+S_P+PH$ for $P\in \Gamma_X$, $H\in \ker\varphi_{X_*}$ and $Q\in O_m(\mathbb{F}_q)$. Then we have
    \[
    X'(X')^T=PX_*X_*^TP^T=-PA_0P^T=-A_0
    \]
    and by Remark~\ref{conrem},
    \begin{align*}
    \varphi_{X'}(Y')&=\varphi_{PX_*}(PY_*)+\varphi_{PX_*}(PH)+\varphi_{PX_*}(S_P)\\&=P(\varphi_{X_*}(Y_*))P^T+P(\varphi_{X_*}(H))P^T+PA_1P^T-A_1\\&=-PA_1P^T+PA_1P^T-A_1\\&=-A_1.
    \end{align*}
    Since the right multiplication by $Q\in O_m(\mathbb{F}_q)$ preserves these calculations,
    \[
    G'=[G~|~(PX_*+u(PY_*+S_P+PH))Q]
    \]
    generates a shortest self-orthogonal embedding of $\mathcal{C}$.

    Conversely, let $G''=[G~|~B'']$ where $B''=X+uY$ generates a shortest self-orthogonal embedding of $\mathcal{C}$. By Theorem~\ref{everychi}, there are $P\in \Gamma_X$ and $Q\in O_m(\mathbb{F}_q)$ such that $X=PX_*Q$. Put
    \[
    Z=YQ^T-PY_*-S_P.
    \]
    Then we have
    \begin{align*}
    \varphi_{PX_*}(Z)&=\varphi_{PX_*}(YQ^T)-\varphi_{PX_*}(PY_*)-\varphi_{PX_*}(S_P)\\&=(PX_*)(YQ^T)^T+(YQ^T)(PX_*)^T+PA_1P^T-PA_1P^T+A_1\\&=(XQ^T)(YQ^T)^T+(YQ^T)(XQ^T)^T+A_1\\
    &=XY^T+YX^T+A_1\\
    &=\mathcal{O}.
    \end{align*}
    Let $H=P^{-1}Z$. This gives
    \[
    P\varphi_{X_*}(H)P^T=\varphi_{PX_*}(PH)=\varphi_{PX_*}(Z)=\mathcal{O}.
    \]
    Then $Y=(PY_*+S_P+PH)Q$. This completes the proof.
\end{proof}

\section{Shortest LCD embeddings of linear codes over {\it R}}\label{sec6}
Throughout this section, we denote $R=\mathbb{F}_q+u\mathbb{F}_q$.

\begin{theorem}\label{LCD-exact}
    Let $\mathcal{C}$ be a code of length $n$ over $R=\mathbb{F}_q+u\mathbb{F}_q$ for some prime power $q$ with generator matrix $G$. Write $GG^T=A_0+uA_1$ for some matrices $A_0, A_1$ over $\mathbb{F}_q$. Let $r=\operatorname{rank}_{\mathbb{F}_q}(A_0)$, and let $k$ be the number of rows of $G$. Then
    \[
    n_\ell(\mathcal{C})=n+k-r.
    \]
\end{theorem}
\begin{proof}
    We first show that $n_\ell(\mathcal{C})\ge n+k-r$. Let $n_\ell(\mathcal{C})-n=m$. Then there is a $k\times m$ matrix $B=B_0+uB_1$ over $R$ with $B_0,\,B_1\in M_{k\times m}(\mathbb{F}_q)$ such that $GG^T+BB^T$ is invertible over $R$. Note that
    \[
    \overline{(GG^T+BB^T)}=A_0+B_0B_0^T
    \]
    is also invertible over $\mathbb{F}_q$. Then we have
    \[
    k=\operatorname{rank}_{\mathbb{F}_q}(A_0+B_0B_0^T)\le r+\operatorname{rank}_{\mathbb{F}_q}(B_0)\le r+m,
    \]
    which implies
    \[
    m\ge k-r.
    \]

    Next, we show that $n_\ell(\mathcal{C})\le n+k-r$. Since $\operatorname{rank}_{\mathbb{F}_q}(A_0)=r$, we have
    \[
    \dim\operatorname{Rad}(A_0)=k-r.
    \]
    Let $u_{r+1}, \ldots, u_k$ be a basis of $\operatorname{Rad}(A_0)$, and extend it to a basis $u_1, \ldots, u_k$ of $\mathbb{F}_q^k$. Let $P\in GL_k(\mathbb{F}_q)$ be a matrix whose rows are $u_1, \ldots, u_k$. Then we have
    \[
    PA_0P^T=S\oplus \mathcal{O}
    \]
    for some $r\times r$ symmetric matrix $S$ over $\mathbb{F}_q$. Suppose that $S$ is not invertible. Then there is a nonzero vector $\mathbf{x}=(x_1, \ldots, x_r)\in \mathbb{F}_q^r$ such that $\mathbf{x}S=\mathbf{0}$. Then we have
    \[
    \mathbf{0}=(\mathbf{x}, 0, \ldots, 0)(S\oplus \mathcal{O})=(\mathbf{x}, 0, \ldots, 0)PA_0P^T.
    \]
    It follows that $(\mathbf{x}, 0, \ldots, 0)P\in \operatorname{Rad}(A_0)$. Note that
    \[
    (\mathbf{x}, 0, \ldots, 0)P\in \operatorname{span}(u_1, \ldots, u_r).
    \]
    Since $(\mathbf{x}, 0, \ldots, 0)P\in \operatorname{span}(u_1, \ldots, u_r)\cap \operatorname{Rad}(A_0)$, and is nonzero, this is a contradiction. Therefore, $S$ must be invertible.

    Define a matrix $D'$ over $\mathbb{F}_q$ as
    \[
    D'=\begin{pmatrix}
        \mathcal{O}_{r\times (k-r)}\\I_{k-r}
    \end{pmatrix},
    \]
    and let $D=P^{-1}D'$. Then we have
    \[
    P(A_0+DD^T)P^T=PA_0P^T+PDD^TP^T=PA_0P^T+D'(D')^T=S\oplus I_{k-r}.
    \]
    Thus, $A_0+DD^T$ is invertible over $\mathbb{F}_q$, and by viewing $D$ as a matrix over $R$, it follows that $GG^T+DD^T$ is invertible over $R$. Hence, we have $m\le k-r$. This completes the proof.
\end{proof}

To prove the characterization of all shortest LCD embeddings in Theorem~\ref{LCD_class}, we use the following lemma to calculate the determinant of matrices.
\begin{lemma}[\cite{ZS-2006}]\label{Schur_lemma}
Let $A$, $B$, $C$, and $D$ be matrices of sizes $n\times n$, $n\times m$, $m\times n$, and $m\times m$, respectively, over a field $F$. If $\det(A)\ne 0$, then the determinant $\det(M)$ of the matrix
\[
M=\begin{bmatrix}
    A & B\\C & D
\end{bmatrix}
\]
is given by $\det (M)=\det (A)\cdot \det(D-CA^{-1}B)$.
\end{lemma}

\begin{theorem}\label{LCD_class}
    Let $\mathcal{C}$ be a code over $R=\mathbb{F}_q+u\mathbb{F}_q$ for some prime power $q$ with generator matrix $G$. Let $GG^T=A_0+uA_1$ for some $A_0, A_1\in M_k(\mathbb{F}_q)$. Let $\ell=\dim\operatorname{Rad}(A_0)$. Let $U$ be a $k\times k$ invertible matrix over $\mathbb{F}_q$ whose first $\ell$ rows form a basis of $\operatorname{Rad}(A_0)$. Then $\tilde{\mathcal{C}}$ is a shortest LCD embedding of $\mathcal{C}$ if and only if it has a generator matrix
    \[
    \tilde{G}=[G~|~B]
    \]
    where
    \[
    B=U^{-1}\begin{pmatrix}
        D\\E
    \end{pmatrix}
    \]
    for $D\in GL_\ell(R)$ and some matrix $E$.
\end{theorem}
\begin{proof}
    Note that $UA_0U^T=\mathcal{O}\oplus S$ for some invertible matrix $S$. Let $D$ be an $\ell\times \ell$ invertible matrix over $R$ and $E$ be any matrix of corresponding size. Then we have
    \[
    \overline{U\tilde{G}\tilde{G}^TU^T}=\overline{U(GG^T+BB^T)U^T}=\begin{pmatrix}
        \overline{D}\,\overline{D}^T & \overline{D}\,\overline{E}^T\\\overline{E}\,\overline{D}^T & \overline{E}\,\overline{E}^T+S
    \end{pmatrix}.
    \]
    By Lemma~\ref{Schur_lemma},
    \begin{align*}
    \det\left(\overline{\tilde{G}\tilde{G}^T}\right)
    &=\det\left(\overline{U\tilde{G}\tilde{G}^TU^T}\right)/\det(\overline{U})^2\\
    &=\det\left(\overline{D}\,\overline{D}^T\right)\det\left(\overline{E}\,\overline{E}^T+S-(\overline{E}\,\overline{D}^T)(\overline{D}\,\overline{D}^T)^{-1}(\overline{D}\,\overline{E}^T)\right)/\det(\overline{U})^2\\
    &=\det\left(\overline{D}\,\overline{D}^T\right)\det(S)/\det(\overline{U})^2\\
    &\ne 0.
    \end{align*}
    Thus, $\tilde{G}$ generates a shortest LCD embedding of $\mathcal{C}$.

    Conversely, suppose that $G'=[G~|~B']$ generates a shortest LCD embedding of $\mathcal{C}$. Let
    \[
    UB'=\begin{pmatrix}
        D'\\E'
    \end{pmatrix}
    \]
    where the number of rows of $D'$ is $\ell$. Suppose that $D'$ is not invertible. Choose a nonzero vector $\mathbf{u}\in \mathbb{F}_q^\ell$ such that $\mathbf{u}\overline{D'}=\mathbf{0}$. Let $\mathbf{u}'=(\mathbf{u}, \mathbf{0})\in \mathbb{F}_q^k$. Then
    \[
    \mathbf{u}'\overline{UG'(G')^TU^T}=(\mathbf{u}, \mathbf{0})\begin{pmatrix}
        \overline{D'}\,\overline{D'}^T & \overline{D'}\,\overline{E'}^T\\\overline{E'}\,\overline{D'}^T & \overline{E'}\,\overline{E'}^T+S
    \end{pmatrix}=\mathcal{O}.
    \]
    Thus, $\mathbf{u}'\in \ker\left(\overline{UG'(G')^TU^T}\right)$ which is a contradiction since $\overline{UG'(G')^TU^T}$ is invertible. Thus $D'$ must be an invertible matrix. This completes the proof.
\end{proof}

\section{Some examples}
In this section we provide examples of self-orthogonal and LCD embeddings of codes over $R$. Throughout this section, we use the same notation as in Sections~\ref{sec3} to~\ref{sec6}.

Examples~\ref{even-case1} and~\ref{even-case2} present shortest self-orthogonal embeddings of codes over rings of even characteristic corresponding to Cases 1 and 2, respectively. Examples~\ref{odd-square} and~\ref{odd-nonsquare} present shortest self-orthogonal embeddings over rings of odd characteristic for the cases where $(-1)^{r+m}\det(S)$ is a square and a nonsquare, respectively. Finally, Example~\ref{LCDex} provides an example of a shortest LCD embedding. Since the Gray map preserves orthogonality over $R=\mathbb{F}_q+u\mathbb{F}_q$ when $q$ is even, we consider the Gray images of the codes only in Examples~\ref{even-case1},~\ref{even-case2} and~\ref{LCDex}. The optimality of these codes can be confirmed by~\cite{Grassl2026}.

Furthermore, using MAGMA, we enumerate all shortest self-orthogonal and shortest LCD embeddings of the codes considered in these examples by applying Theorems~\ref{optimalconstruct} and~\ref{LCD_class}, respectively, and provide those having the largest minimum distance.

\begin{example}\label{even-case1}
    Let $\mathcal{C}$ be a code of length $n=20$ over $R=\mathbb{F}_2+u\mathbb{F}_2$ with generator matrix
    \[
    \begingroup
    \setlength{\arraycolsep}{2pt}
    \tiny
    G=\left[
    \begin{array}{cccccccccccccccccccc}
    1&0&1&0&1&1+u&1&0&1&0&0&1&1+u&u&u&1+u&1+u&0&1&1\\
    1+u&1&1&0&1&0&0&1&1&1&0&u&1+u&u&0&1+u&1&0&u&0\\
    0&1+u&0&1&0&1&0&0&1&0&1+u&u&u&1&1+u&u&1&0&1+u&1\\
    0&0&u&0&1&1+u&1&1&0&0&1&1+u&0&1&u&1+u&1+u&1+u&1&1+u\\
    1&1&1&1&1&1+u&1&0&1&0&u&1&0&1+u&1&0&1&1&1&u\\
    1&0&0&0&0&0&1&1&1&0&0&0&u&u&1+u&1+u&1+u&1+u&1+u&1+u\\
    0&0&0&0&0&u&0&0&0&u&0&0&0&u&0&0&u&0&u&u\\
    u&0&0&0&0&0&0&u&0&0&u&u&0&0&0&u&0&0&u&0
    \end{array}
    \right].
    \endgroup
    \]
    The code $\mathcal{C}$ is of type $\{6, 2\}$ with minimum distance 6. Let $\phi:R^{20}\to \mathbb{F}_2^{40}$ be the Gray map. Then $\phi(\mathcal{C})$ is a $[40, 14, 12]$ code over $\mathbb{F}_2$.
    
    Let $GG^T=A_0+uA_1$ for $A_0,\,A_1\in M_{8\times 8}(\mathbb{F}_2)$. Then $A_0$ is given as
    \[
    A_0=\left[
    \begin{array}{c|c}
        \begin{matrix}
            0 & 1 & 1 & 0 & 1 & 1\\
            1 & 0 & 1 & 0 & 0 & 1\\
            1 & 1 & 0 & 0 & 0 & 1\\
            0 & 0 & 0 & 0 & 0 & 1\\
            1 & 0 & 0 & 0 & 0 & 1\\
            1 & 1 & 1 & 1 & 1 & 0
        \end{matrix} & \mathcal{O}\\\hline
        \mathcal{O} & \mathcal{O}
    \end{array}
    \right].
    \]
    Let $\bar{A}_0$ be the matrix obtained by reducing the top-left $6\times 6$ submatrix of $A_0$ modulo $u$. Since $\bar{A}_0$ is nonzero and all its diagonal entries are zero, $\operatorname{Res}(\mathcal{C})$ is an even code that is not self-orthogonal. 
    
    Let $U\in GL_{8}(R)$ be the matrix given as
    \[
    U=\begin{bmatrix}
    1 & 0 & 0 & 0 & 0 & 0 & 0 & 0\\
    u & 1+u & 0 & 0 & 0 & u & 0 & 0\\
    1 & 1+u & 1+u & 0 & u & u & 0 & 0\\
    u & 1 & 0 & 0 & 1+u & u & 0 & 0\\
    0 & 0 & 0 & 1+u & 0 & 0 & 0 & 0\\
    1 & 0 & 0 & 0 & 1 & 1 & 0 & 0\\
    u & u & 0 & u & 0 & u & 1 & 0\\
    0 & u & 0 & u & 0 & u & 0 & 1
    \end{bmatrix}.
    \]
    Then we have
    \[
    UGG^TU^T=J_6\oplus u\mathcal{O}_2.
    \]
    Since the bottom-right $2\times 2$ submatrix of $UGG^TU^T$ is the zero matrix, we have $\rho=0$. So, the code $\mathcal{C}$ is a code of Case I and by Theorem~\ref{thm:exact_even}, the length of a shortest self-orthogonal embedding of $\mathcal{C}$ is given as
    \[
    n_s(\mathcal{C})=n+\operatorname{rank}_{\mathbb{F}_q}(A_0)+1=20+6+1=27.
    \]
    Following the construction given in the proof of Lemma~\ref{nuzerodiag}, choose the matrix $\tilde{B}$ as 
    \[
    \tilde{B}=\begin{bmatrix}
    1 & 0 & 0 & 0 & 0 & 0 & 1\\
    0 & 1 & 0 & 0 & 0 & 0 & 1\\
    1 & 1 & 1 & 0 & 0 & 0 & 1\\
    1 & 1 & 0 & 1 & 0 & 0 & 1\\
    1 & 1 & 1 & 1 & 1 & 0 & 1\\
    1 & 1 & 1 & 1 & 0 & 1 & 1\\
    0 & 0 & 0 & 0 & 0 & 0 & 0\\
    0 & 0 & 0 & 0 & 0 & 0 & 0
    \end{bmatrix}.
    \]
    Since $J_6\oplus u\mathcal{O}_2+\tilde{B}\tilde{B}^T=\mathcal{O}$, the code generated by
    \[
    [G~|~U^{-1}\tilde{B}]
    \]
    generates a shortest self-orthogonal embedding $\mathcal{C}'$ of $\mathcal{C}$, and its minimum distance is 8, and $\phi(\mathcal{C}')$ is a $[54, 14, 14]$ code over $\mathbb{F}_2$.
    
    Set $B=U^{-1}\tilde{B}$. Choose
    \[
    P'=I_8\in \Gamma_X,\quad S_{P'}=\mathcal{O}_{8\times 7}
    \]
    where $B=X+uY$ and
    \[
    H=\begin{bmatrix}
    1 & 0 & 0 & 1 & 0 & 1 & 0\\
    1 & 0 & 1 & 1 & 1 & 1 & 1\\
    0 & 0 & 1 & 1 & 1 & 0 & 0\\
    1 & 0 & 0 & 1 & 1 & 0 & 1\\
    1 & 1 & 0 & 0 & 1 & 1 & 1\\
    1 & 0 & 0 & 0 & 1 & 1 & 1\\
    0 & 0 & 0 & 0 & 0 & 0 & 0\\
    0 & 0 & 0 & 0 & 0 & 0 & 0
    \end{bmatrix}\in \ker\varphi_{B},\quad\mbox{and}\quad
    Q=\begin{bmatrix}
    1 & 1 & 0 & 1 & 1 & 0 & 1\\
    0 & 1 & 1 & 1 & 0 & 0 & 0\\
    0 & 1 & 1 & 0 & 0 & 0 & 1\\
    1 & 0 & 1 & 1 & 1 & 0 & 1\\
    0 & 0 & 0 & 0 & 0 & 1 & 0\\
    1 & 1 & 1 & 0 & 0 & 0 & 0\\
    0 & 1 & 1 & 0 & 1 & 0 & 0
    \end{bmatrix}.
    \]
    By Theorem~\ref{optimalconstruct}, the code $\tilde{\mathcal{C}}$ generated by
    \[
    \tilde{G}=[G~|~(P'B+u(P'\mathcal{O}+S_{P'}+P'H))Q]
    \]
    is a shortest self-orthogonal embedding of $\mathcal{C}$
    where
    \[
    (P'B+u(P'\mathcal{O}+S_{P'}+P'H))Q=
    \begin{bmatrix}
    1+u & 0 & 1 & 1 & 0 & 0 & 1\\
    0 & 0 & 0 & 1 & 1+u & u & 0\\
    0 & 0 & 0 & 0 & 1 & u & 1\\
    0 & 0 & 0 & 1+u & 1 & 1 & 1+u\\
    0 & 1+u & 1 & u & u & u & u\\
    0 & 0 & 1+u & 0 & 1+u & u & u\\
    u & 0 & 0 & u & u & u & 0\\
    0 & 0 & u & 0 & u & u & u
    \end{bmatrix}.
    \]
    The minimum distance of $\tilde{\mathcal{C}}$ is $10$, and we have checked that this is the largest minimum distance among all shortest self-orthogonal embeddings of $\mathcal{C}$. Moreover, $\phi$ maps $\tilde{\mathcal{C}}$ into a $[54, 14, 20]$ code over $\mathbb{F}_2$ which is an optimal self-orthogonal code.
\end{example}

\begin{example}\label{even-case2}
    Let $\mathcal{C}$ be a code of length $n=22$ over $R=\mathbb{F}_2+u\mathbb{F}_2$ with generator matrix
    \[
    \begingroup
    \setlength{\arraycolsep}{2pt}
    \tiny
    G=
    \left[
    \begin{array}{cccccccccccccccccccccc}
    0&1&0&1&1&1+u&1&0&1&0&1&0&1+u&0&u&u&1+u&1+u&0&1&1&1\\
    1&1&0&0&1&0&0&1&1&1&u&1+u&1+u&u&u&0&1+u&1&0&u&0&0\\
    1+u&0&1&0&0&1&0&0&1&0&u&1&u&u&1&1+u&u&1&0&1+u&1&1\\
    0&u&0&0&1&1+u&1&1&0&0&1+u&1&0&1&1&u&1+u&1+u&1+u&1&1+u&1+u\\
    1&1&1&1&1&1+u&1&0&1&0&1&u&0&u&1+u&1&0&1&1&1&u&u\\
    0&0&0&1+u&0&0&1&1&1&0&0&1+u&u&u&u&1+u&1+u&1+u&1+u&1+u&u&1+u\\
    0&0&0&0&0&u&0&0&0&u&0&u&0&u&u&0&0&u&0&u&0&u\\
    0&0&0&u&0&0&0&u&0&0&u&u&0&u&0&0&u&0&0&u&u&0
    \end{array}
    \right].
    \endgroup
    \]
    The code $\mathcal{C}$ is of type $\{6, 2\}$ with minimum distance 7. Let $\phi:R^{22}\to \mathbb{F}_2^{44}$ be the Gray map. Then $\phi(\mathcal{C})$ is a $[44, 14, 13]$ code over $\mathbb{F}_2$.
    
    Let $GG^T=A_0+uA_1$ for $A_0,\,A_1\in M_{8\times 8}(\mathbb{F}_2)$. Then $A_0$ is given as
    \[
    A_0=\left[
    \begin{array}{c|c}
        \begin{matrix}
            1 & 0 & 0 & 1 & 1 & 1\\
            0 & 0 & 0 & 1 & 1 & 1\\
            0 & 0 & 1 & 1 & 0 & 0\\
            1 & 1 & 1 & 0 & 0 & 0\\
            1 & 1 & 0 & 0 & 0 & 1\\
            1 & 1 & 0 & 0 & 1 & 1\\
        \end{matrix} & \mathcal{O}\\\hline
        \mathcal{O} & \mathcal{O}
    \end{array}
    \right].
    \]
    Let $\bar{A}_0$ be the matrix obtained by reducing the top-left $6\times 6$ submatrix of $A_0$ modulo $u$. Since $\bar{A}_0$ has a nonzero diagonal entry, $\operatorname{Res}(\mathcal{C})$ is not an even code.
    
    Let $U\in GL_{8}(R)$ be the matrix given as
    \[
    U=\begin{bmatrix}
    1+u & 0 & 0 & 0 & u & 0 & 0 & 0\\
    0 & u & 1 & 0 & 0 & 0 & 0 & 0\\
    1+u & u & u & u & 1 & 0 & 0 & 0\\
    1 & 1 & 0 & 0 & 1 & 0 & 0 & 0\\
    0 & 0 & 1 & 1 & 1 & 0 & 0 & 0\\
    0 & 1+u & 1+u & 1 & 0 & 1 & 0 & 0\\
    u & u & u & u & 0 & 0 & 1 & 0\\
    0 & u & 0 & u & 0 & 0 & 0 & 1
    \end{bmatrix}.
    \]
    Then we have
    \[
    UGG^TU^T=I_5\oplus u\mathcal{O}_3.
    \]
    Since the bottom-right $3\times 3$ submatrix of $UGG^TU^T$ is the zero matrix, we have $\rho=0$. So, by Theorem~\ref{thm:exact_even}, the length of a shortest self-orthogonal embedding of $\mathcal{C}$ is given as
    \[
    n_s(\mathcal{C})=n+\operatorname{rank}_{\mathbb{F}_q}(A_0)+\rho=22+5+0=27.
    \]
    Choose the matrix $\tilde{B}$ as 
    \[
    \tilde{B} = \begin{bmatrix}
        I_5\\\mathcal{O}_{3\times 5}
    \end{bmatrix}.
    \]
    Since $I_5\oplus u\mathcal{O}_3+\tilde{B}\tilde{B}^T=\mathcal{O}$, the code generated by
    \[
    [G~|~U^{-1}\tilde{B}]
    \]
    generates a shortest self-orthogonal embedding $\mathcal{C}'$ of $\mathcal{C}$, and its minimum distance is 8. Also, $\phi(\mathcal{C}')$ is a $[54, 14, 16]$ code over $\mathbb{F}_2$.

    Set $B=U^{-1}\tilde{B}$. Choose
    \[
    P'=I_8\in \Gamma_X,\quad S_{P'}=\mathcal{O}_{8\times 5}
    \]
    where $B=X+uY$ and
    \[
    H=\begin{bmatrix}
    1 & 0 & 0 & 1 & 0\\
    1 & 0 & 0 & 0 & 0\\
    0 & 1 & 1 & 1 & 0\\
    1 & 0 & 1 & 0 & 0\\
    1 & 1 & 1 & 0 & 1\\
    0 & 1 & 0 & 1 & 0\\
    0 & 0 & 0 & 0 & 0\\
    0 & 0 & 0 & 0 & 0
    \end{bmatrix}\in \ker\varphi_{B},\quad\mbox{and}\quad
    Q=\begin{bmatrix}
    1 & 1 & 0 & 1 & 0\\
    0 & 0 & 0 & 0 & 1\\
    1 & 0 & 1 & 1 & 0\\
    1 & 1 & 1 & 0 & 0\\
    0 & 1 & 1 & 1 & 0
    \end{bmatrix}.
    \]
    By Theorem~\ref{optimalconstruct}, the code $\tilde{\mathcal{C}}$ generated by
    \[
    \tilde{G}=[G~|~(P'B+u(P'\mathcal{O}+S_{P'}+P'H))Q]
    \]
    is a shortest self-orthogonal embedding of $\mathcal{C}$
    where
    \[
    (P'B+u(P'\mathcal{O}+S_{P'}+P'H))Q=
    \begin{bmatrix}
    1+u & 1 & 0 & 1 & 0\\
    0 & 1+u & 0 & 1 & 0\\
    0 & 0 & 0 & 0 & 1+u\\
    0 & 0 & 0 & 1+u & 1\\
    0 & 1 & 1+u & u & u\\
    0 & 1 & 0 & 0 & 0\\
    u & 0 & 0 & u & 0\\
    0 & u & 0 & 0 & u
    \end{bmatrix}.
    \]
    The minimum distance of $\tilde{\mathcal{C}}$ is $10$, and we have checked that this is the largest minimum distance among all shortest self-orthogonal embeddings of $\mathcal{C}$. Moreover, $\phi$ maps $\tilde{\mathcal{C}}$ into a $[54, 14, 20]$ code over $\mathbb{F}_2$ which is an optimal self-orthogonal code. In particular, the shortest self-orthogonal embeddings attaining the largest possible minimum distance for the codes in Examples~\ref{even-case1} and~\ref{even-case2} coincide.
\end{example}

\begin{example}\label{odd-square}
    Let $\mathcal{C}$ be a code of length $n=23$ over $R=\mathbb{F}_3+u\mathbb{F}_3$ with generator matrix
    \[
    \begingroup
    \setlength{\arraycolsep}{2pt}
    \tiny
    G=
    \left[
    \begin{array}{ccccccccccccccccccccccc}
    1&2&0&2&1&0&0&0&0&1&2&2&0&2&1&0&2&2&2&1&1&2&2\\
    0&1&2&0&2&1&0&0&0&1&2&0&1&0&2&1&1&0&1&1&2&0&0\\
    0&0&1&2&0&2&1&0&0&1&1&1&0&1&0&2&1&1&1&2&2&2&0\\
    0&0&0&1&2&0&2&1&0&2&2&1&1&2&2&2&2&1&1&0&1&2&2\\
    0&0&0&0&1&2&0&2&1&1&1&2&2&2&0&1&1&2&0&2&2&0&2\\
    2u&u&0&0&0&0&0&0&0&u&2u&u&u&0&0&u&2u&0&0&2u&0&0&2u\\
    0&2u&u&0&0&0&0&0&0&u&0&u&0&0&u&u&u&2u&0&0&0&2u&0\\
    0&0&2u&u&0&0&0&0&0&u&2u&0&0&0&0&2u&u&2u&2u&0&2u&2u&2u
    \end{array}
    \right].
    \endgroup
    \]
    The code $\mathcal{C}$ is of type $\{5, 3\}$ with minimum distance 8. Let $GG^T=A_0+uA_1$ for $A_0,\,A_1\in M_{8\times 8}(\mathbb{F}_3)$. Then $A_0$ is given as
    \[
    A_0=\left[
    \begin{array}{c|c}
        \begin{matrix}
           1 & 0 & 1 & 2 & 2\\
           0 & 1 & 2 & 1 & 2\\
           1 & 2 & 0 & 1 & 2\\
           2 & 1 & 1 & 2 & 1\\
           2 & 2 & 2 & 1 & 0
        \end{matrix} & \mathcal{O}\\\hline
        \mathcal{O} & \mathcal{O}
    \end{array}
    \right].
    \]
    Let $\bar{A}_0$ be the matrix obtained by reducing the top-left $5\times 5$ submatrix of $A_0$ modulo $u$. Note that $\det\bar{A}_0=2\in \mathbb{F}_3^\times\backslash(\mathbb{F}_3^\times)^2$.
    
    Let
    \[
    P=\begin{bmatrix}
        1 & 0 & 0 & 0 & 0\\
        0 & 1 & 0 & 0 & 0\\
        2 & 1 & 1 & 0 & 0\\
        2 & 0 & 1 & 2 & 1\\
        1 & 1 & 1 & 1 & 1
    \end{bmatrix}\oplus I_3\in GL_8(\mathbb{F}_3).
    \]
    Then we have
    \[
    PGG^TP^T=I_4\oplus[2]\oplus u\mathcal{O}_3.
    \]
    Since the bottom-right $3\times 3$ submatrix of $PGG^TP^T$ is the zero matrix, we have $m=\mu(-Z)=0$, and 
    \[
    (-1)^{r+m}\det(I_4\oplus[2])=-2=1\in (\mathbb{F}_3^\times)^2
    \]
    where $r=\operatorname{rank}_{\mathbb{F}_q}(A_0)$. Thus, by Theorem~\ref{thm:exact_odd} the length of a shortest self-orthogonal embedding of $\mathcal{C}$ is given as
    \[
    n_s(\mathcal{C})=n+r+2m=23+5+0=28.
    \]
    Choose the matrix $B$ as 
    \[
    B = \begin{bmatrix}
        1 & 1 & 1 & 1 & 1\\
        0 & 1 & 0 & 0 & 2\\
        2 & 0 & 0 & 1 & 0\\
        2 & 2 & 1 & 2 & 2\\
        1 & 2 & 0 & 1 & 2\\
        0 & 0 & 0 & 0 & 0\\
        0 & 0 & 0 & 0 & 0\\
        0 & 0 & 0 & 0 & 0
    \end{bmatrix}.
    \]
    Since $I_4\oplus[2]\oplus u\mathcal{O}_3+BB^T=\mathcal{O}$, the code generated by
    \[
    [G~|~P^{-1}B]
    \]
    generates a shortest self-orthogonal embedding of $\mathcal{C}$, and its minimum distance is 9. We have checked that this is the largest minimum distance among all shortest self-orthogonal embeddings of $\mathcal{C}$.  
\end{example}

\begin{example}\label{odd-nonsquare}
    Let $\mathcal{C}$ be a code of length $n=22$ over $R=\mathbb{F}_3+u\mathbb{F}_3$ with generator matrix
    \[
    \begingroup
    \setlength{\arraycolsep}{2pt}
    \tiny
    G=
    \left[
    \begin{array}{cccccccccccccccccccccc}
    2&2&2&1&1&1&0&0&0&0&2&2&0&1&0&0&0&1&2&2&1&2\\
    0&2&2&2&1&1&1&0&0&2&2&2&1&1&1&1&1&1&1&0&2&0\\
    0&0&2&2&2&1&1&1&0&2&2&0&2&0&1&0&0&2&1&2&0&0\\
    0&0&0&2&2&2&1&1&1&1&1&0&2&0&2&1&1&1&1&2&1&1\\
    2u&u&0&0&0&0&0&0&0&u&0&0&2u&2u&2u&u&0&0&u&0&0&u\\
    0&2u&u&0&0&0&0&0&0&2u&0&u&0&u&u&2u&2u&2u&u&0&0&u\\
    0&0&2u&u&0&0&0&0&0&u&2u&2u&2u&2u&2u&2u&2u&0&0&0&2u&2u\\
    0&0&0&2u&u&0&0&0&0&u&0&u&2u&2u&2u&0&0&0&2u&u&2u&0
    \end{array}
    \right].
    \endgroup
    \]
    The code $\mathcal{C}$ is of type $\{4, 4\}$ with minimum distance 7. Let $GG^T=A_0+uA_1$ for $A_0,\,A_1\in M_{8\times 8}(\mathbb{F}_3)$. Then $A_0$ is given as
    \[
    A_0=\left[
    \begin{array}{c|c}
        \begin{matrix}
           2 & 2 & 0 & 0\\
           2 & 2 & 2 & 2\\
           0 & 2 & 1 & 2\\
           0 & 2 & 2 & 2
        \end{matrix} & \mathcal{O}\\\hline
        \mathcal{O} & \mathcal{O}
    \end{array}
    \right].
    \]
    Let $\bar{A}_0$ be the matrix obtained by reducing the top-left $4\times 4$ submatrix of $A_0$ modulo $u$. Note that $\det\bar{A}_0=2\in \mathbb{F}_3^\times\backslash(\mathbb{F}_3^\times)^2$.
    
    Let
    \[
    U=\begin{bmatrix}
        0 & 0 & 1 & 0\\
        0 & 1 & 1 & 0\\
        1 & 1 & 1 & 0\\
        2 & 1 & 2 & 1
    \end{bmatrix}\oplus I_4\in GL_8(\mathbb{F}_3).
    \]
    Then we have
    \[
    UGG^TU^T=I_3\oplus[2]\oplus u\mathcal{O}_4.
    \]
    Since the bottom-right $4\times 4$ submatrix of $UGG^TU^T$ is the zero matrix, we have $m=\mu(-Z)=0$, and 
    \[
    (-1)^{r+m}\det(I_3\oplus[2])=2\notin (\mathbb{F}_3^\times)^2
    \]
    where $r=\operatorname{rank}_{\mathbb{F}_q}(A_0)$. Thus, by Theorem~\ref{thm:exact_odd} the length of a shortest self-orthogonal embedding of $\mathcal{C}$ is given as
    \[
    n_s(\mathcal{C})=n+r+2m+1=22+4+0+1=27.
    \]
    Choose the matrix $\tilde{B}$ as 
    \[
    \tilde{B} = \begin{bmatrix}
        0 & 0 & 0 & 1 & 1\\
        1 & 1 & 1 & 1 & 2\\
        1 & 1 & 2 & 1 & 2\\
        1 & 1 & 0 & 2 & 1\\
        0 & 0 & 0 & 0 & 0\\
        0 & 0 & 0 & 0 & 0\\
        0 & 0 & 0 & 0 & 0\\
        0 & 0 & 0 & 0 & 0
    \end{bmatrix}.
    \]
    Since $I_3\oplus[2]\oplus u\mathcal{O}_4+\tilde{B}\tilde{B}^T=\mathcal{O}$, the code generated by
    \[
    [G~|~U^{-1}\tilde{B}]
    \]
    generates a shortest self-orthogonal embedding of $\mathcal{C}$, and its minimum distance is 8.

    Set $B=U^{-1}\tilde{B}$. Choose $P'=I_8\in \Gamma_X$, $S_{P'}=\mathcal{O}_{8\times 5}$ where $B=X+uY$ and
    \[
    H=\begin{bmatrix}
        0 & 2 & 0 & 0 & 0\\
        0 & 0 & 1 & 0 & 2\\
        2 & 1 & 0 & 2 & 1\\
        2 & 2 & 0 & 2 & 0\\
        2 & 1 & 0 & 0 & 0\\
        2 & 1 & 0 & 0 & 0\\
        2 & 1 & 0 & 0 & 0\\
        0 & 0 & 0 & 0 & 0
    \end{bmatrix}\in \ker\varphi_{B}.
    \]
    By Theorem~\ref{optimalconstruct}, the code $\tilde{\mathcal{C}}$ generated by
    \[
    \tilde{G}=[G~|~P'B+u(P'\mathcal{O}+S_{P'}+P'H)]
    \]
    is a shortest self-orthogonal embedding of $\mathcal{C}$
    where
    \[
    P'B+u(P'\mathcal{O}+S_{P'}+P'H)=\begin{bmatrix}
        0 & 2u & 1 & 0 & 0\\
        1 & 1 & 1+u & 0 & 1+2u\\
        2u & u & 0 & 1+2u & 1+u\\
        2u & 2u & 0 & 2u & 1\\
        2u & u & 0 & 0 & 0\\
        2u & u & 0 & 0 & 0\\
        2u & u & 0 & 0 & 0\\
        0 & 0 & 0 & 0 & 0
    \end{bmatrix}.
    \]
    The minimum distance of $\tilde{\mathcal{C}}$ is $10$, and we have checked that this is the largest minimum distance among all shortest self-orthogonal embeddings of $\mathcal{C}$.
\end{example}

\begin{example}\label{LCDex}
    Let $\mathcal{C}$ be a code of length $n=8$ over $R=\mathbb{F}_2+u\mathbb{F}_2$ with generator matrix
    \[
    G=
    \left[
    \begin{array}{cccccccc}
    1 & 0 & 0 & 1 & u & u & 0 & 1\\
    0 & 1 & 0 & 1+u & 1 & 1+u & 1+u & 1\\
    0 & 0 & 1 & 1+u & 1+u & 0 & u & u
    \end{array}
    \right].
    \]
    The code $\mathcal{C}$ is of type $\{3, 0\}$ with minimum distance 3. Let $\phi:R^{8}\to \mathbb{F}_2^{16}$ be the Gray map. Then $\phi(\mathcal{C})$ is a $[16, 6, 6]$ code over $\mathbb{F}_2$.
    
    Let $GG^T=A_0+uA_1$ for $A_0,\,A_1\in M_{3\times 3}(\mathbb{F}_2)$. Then
    \[
    A_0=\begin{bmatrix}
        1 & 0 & 1\\
        0 & 0 & 0\\
        1 & 0 & 1
    \end{bmatrix}.
    \]
    Thus, the length of a shortest LCD embedding of $\mathcal{C}$ is given as
    \[
    n_\ell(\mathcal{C})=n+k-\operatorname{rank}_{\mathbb{F}_q}(A_0)=8+3-1=10
    \]
    by Theorem~\ref{LCD-exact}.

    Take a matrix $U$ as
    \[
    U=\begin{bmatrix}
        0 & 1 & 0\\
        1 & 0 & 1\\
        0 & 0 & 1
    \end{bmatrix}.
    \]
    Then the first two rows of $U$ form a basis of $\operatorname{Rad}(A_0)$. Let
    \[
    D=\begin{bmatrix}
        u & 1+u\\
        1+u & 1+u
    \end{bmatrix}\in GL_2(R)\quad\mbox{and}\quad
    E=\begin{bmatrix}
        1+u & u
    \end{bmatrix}.
    \]
    Define
    \[
    B=U^{-1}\begin{bmatrix}
        D\\E
    \end{bmatrix}=\begin{bmatrix}
        0 & 1\\
        u & 1+u\\1+u & u
    \end{bmatrix}.
    \]
    By Theorem~\ref{LCD_class}, the code $\tilde{\mathcal{C}}$ generated by $[G~|~B]$ is a shortest LCD embedding of $\mathcal{C}$, and its minimum distance is $4$. We have checked that this is the largest minimum distance among all shortest LCD embeddings of $\mathcal{C}$. Furthermore, $\phi$ maps $\tilde{\mathcal{C}}$ into a $[20, 6, 8]$ code over $\mathbb{F}_2$ which is an optimal LCD code.
\end{example}

\section{Conclusion}
In this paper, we determined the exact lengths of shortest self-orthogonal and LCD embeddings of linear codes over $R=\mathbb{F}_q+u\mathbb{F}_q$. For shortest self-orthogonal embeddings, we obtained complete results covering two cases in each of even and odd characteristic and developed a method for constructing all shortest embeddings. We illustrated the four cases with examples, enumerated all shortest self-orthogonal embeddings in each example, and identified one with the largest minimum distance. Moreover, in some of these examples, the Gray images of the selected embeddings are optimal codes over $\mathbb{F}_q$. We also characterized all shortest LCD embeddings by showing that they can be obtained by appending an arbitrary invertible matrix together with an arbitrary matrix of prescribed sizes to a generator matrix of the given code. We applied this characterization to enumerate all shortest LCD embeddings in an example and identify one having the largest minimum distance.

The embeddings considered in this paper may change the type of the original code. Therefore, a natural direction for future research is to investigate shortest embedding problems over rings under the additional requirement that the type of the original code be preserved.

\section*{Acknowledgement}
This research (J.-L. Kim) was supported in part by the BK21 FOUR (Fostering Outstanding Universities for Research) funded by the Ministry of Education (MOE, Korea), National Research Foundation of Korea (NRF) under Grant No. 4120240415042, National Research Foundation of Korea under Grant No. RS-2024-NR121331, and by Basic Science Research Program through the National Research Foundation of Korea (NRF) funded by the Ministry of Science and ICT under Grant No. RS-2025-24534992 and Global - Learning \& Academic research institution for Master’s·PhD students, and Postdocs(LAMP) Program of the National Research Foundation of Korea(NRF) grant funded by the Ministry of Education(No. RS-2024-00441954).

\end{document}